\documentclass[pra, singlecolumn, superscriptaddress]{revtex4-2}

\usepackage[utf8]{inputenc}
\usepackage[T1]{fontenc}

\usepackage{amsmath,mathtools}
\usepackage{amsfonts}
\usepackage{amssymb}
\usepackage{amsthm}
\usepackage{bm}
\usepackage{dsfont}
\usepackage{physics}

\newtheorem{lemma}{Lemma}
\newtheorem{proposition}{Proposition}

\newtheorem{remark}{Remark}

\usepackage{graphicx}
\usepackage{subfigure}
\usepackage{booktabs}
\usepackage{dcolumn}
\usepackage[labelfont=bf, justification=Justified, format=plain]{caption}

\usepackage{cases}
\usepackage{xcolor}
\usepackage[normalem]{ulem}
\usepackage{tcolorbox}

\definecolor{mygreen}{rgb}{0.01, 0.31, 0.59}
\definecolor{myblue}{rgb}{0.01, 0.31, 0.59}
\usepackage[colorlinks=true]{hyperref}
\hypersetup{linkcolor=myblue, citecolor=mygreen, urlcolor=myblue}

\newcommand{\QFI}{\mathcal{F}}
\newcommand{\CFI}{F}

\newcommand{\sympform}{\mathbf{J}}       
\newcommand{\Lind}{\mathcal{L}}
\newcommand{\diss}[1]{\mathcal{D}[#1]}
\newcommand{\osc}{\Omega}               
\newcommand{\oscR}{\Omega_{\!R}}        
\newcommand{\effW}{W}                   

\begin{document}

\title{Non-equilibrium quantum thermometry with bosonic samples}

\author{Marek Winczewski}
\email{Marek.Winczewski@ug.edu.pl}
\affiliation{%
  Institute of Informatics,
  Faculty of Mathematics, Physics and Informatics,
  University of Gda\'nsk,
  ul.~prof.~Marii Janion~7,
  80-309~Gda\'nsk, Poland}

\author{Micha{\l} Horodecki}
\affiliation{%
  International Centre for Theory of Quantum Technologies (ICTQT),
  University of Gda\'nsk,
  ul.~Wita Stwosza~63,
  80-308~Gda\'nsk, Poland}

\author{Ricard Ravell Rodr\'{\i}guez}
\email{ricard.ra.ro@gmail.com}
\affiliation{%
  ICFO-Institut de Ci\`encies Fot\`oniques,
  The Barcelona Institute of Science and Technology,
  08860~Castelldefels (Barcelona), Spain}

\date{\today}

\begin{abstract}
We study low-temperature non-equilibrium quantum thermometry with a bosonic probe: a quantum
harmonic oscillator strongly coupled to a bosonic bath at temperature $T$ through a
Drude--Ohmic spectral density. We treat the probe--bath dynamics both exactly,
using the quadratic solution of Boyanovsky and Jasnow, and within a renormalized
Gorini--Kossakowski--Lindblad--Sudarshan (GKLS) master equation. From the
time-dependent covariance matrix we extract the quantum Fisher information (QFI)
for general single-mode Gaussian probe states, including squeezed ones. In the
strong-coupling, non-Markovian regime the QFI is non-monotonic in time, displaying
bath-memory revivals that make a finite interrogation time $t^*>0$ strictly
optimal. By contrast, we prove that the Markovian QFI rises monotonically to its
stationary value and develops no interior optimum, so that its optimum is always
pinned to the boundary $t^*\to\infty$; this complements existing Markovian precision-rate bounds, which concern $(\mathcal F(t)/t)$ rather than the single-shot QFI
$(\mathcal F(t))$. Squeezed initial states yield a large transient advantage that
thermalisation eventually erases, establishing squeezing and interrogation time as
complementary thermometric resources. At equilibrium, strong coupling replaces the
exponential Boltzmann suppression of the low-temperature relative error by a far
milder polynomial divergence. As the model maps directly onto circuit quantum
electrodynamics, these protocols appear within current experimental reach.
\end{abstract}

\maketitle

\section{Introduction}
\label{sec:intro}

Estimating the temperature of a physical system is one of the most fundamental
metrological tasks in nature, and becomes particularly demanding when the system of
interest operates deep in the quantum regime~\cite{Mehboudi2019}.
In this context, the precision of any temperature estimate is ultimately bounded
by the laws of quantum mechanics, and the growing field of quantum thermometry is
devoted to characterising these fundamental limits and devising protocols that
approach them~\cite{Mehboudi2019,Paris2009,BraunsteinCaves1994}.
The standard approach casts thermometry as a quantum parameter-estimation
problem~\cite{Paris2009}: a small quantum system, the probe, is prepared in a
suitable state, brought into contact with the sample at temperature $T$, and
subsequently measured; from the measurement outcomes one constructs an estimator
$\tilde T$ whose variance is lower-bounded by the quantum Cram\'er--Rao
inequality~\cite{Cramer1999,Rao1992,BraunsteinCaves1994}
\begin{align}
  \operatorname{Var}(\tilde T) \;\geq\; \frac{1}{\mathcal{F}[\rho_T]},
\end{align}
where $\mathcal{F}[\rho_T]$ is the quantum Fisher information (QFI) of the probe
state $\rho_T$~\cite{Paris2009,BraunsteinCaves1994}.
Maximising the QFI over all measurements and initial probe states therefore sets
the ultimate thermometric precision achievable within a given protocol.

The regime of equilibrium thermometry, in which the probe is measured only after
it has fully thermalised with the sample, is by now theoretically well understood.
The optimal probe structure has been characterised~\cite{Correa2015}, the
low-temperature behaviour has been analysed~\cite{HovhannisyanCorrea2018,Potts2019},
and the role of the system--bath coupling strength has received growing attention.
Especially relevant to the present work are results for bosonic probes coupled to
a bosonic bath through the Caldeira--Leggett model~\cite{CaldeiraLeggett1983}:
moving away from the weak-coupling regime was shown to substantially enhance the
equilibrium QFI at low temperatures~\cite{Correa2017}, with bath-induced
correlations providing an additional resource~\cite{Planella2022} that survives
even when one restricts to energy measurements~\cite{Glatthard2023}.

Beyond equilibrium, non-equilibrium (transient) thermometry considers protocols in
which the probe is measured at a finite interaction time $t$, before
thermalisation is complete.
This regime is richer, but also more demanding, since the precision now depends on
the full non-equilibrium dynamics of the probe, governed by its open-system
coupling to the bath.
In the Markovian limit, fundamental bounds on the achievable precision rate
$\mathcal{F}/t$ have been established~\cite{SekatskiPerarnau2022}, with the optimal
protocol corresponding to a vanishing interrogation time, $t^*\to 0$.
The picture changes dramatically once non-Markovian effects enter.
In the companion paper~\cite{RavellRodriguez2024}, which treated a fermionic probe
strongly coupled to a fermionic bath, the QFI was found to behave non-monotonically
in time, displaying non-Markovian revivals that can be exploited to exceed the
Markovian bound, so that the optimal interrogation time becomes finite, $t^*>0$.
Closer to our setup, Ref.~\cite{Porto2025} studied non-equilibrium thermometry of
a bosonic probe in the quantum Brownian motion framework, albeit at weak coupling
and in specific Ohmic limits.
The advantage of strong coupling is, however, not unconditional.
A recent study of a qubit probe~\cite{Tan2025} finds that, while strong coupling
does lower the steady-state error in the ultralow-temperature limit, weak coupling
can be preferable for the transient signal-to-noise ratio at moderately low
temperatures, so that the regime in which strong coupling actually helps has to be
identified rather than assumed.
The interplay of strong coupling and non-Markovianity in the bosonic case---where
the probe additionally admits squeezed preparations that have no qubit
analogue---thus remains largely open, and is the focus of the present work.

A feature of bosonic probes that has no counterpart for fermionic ones is the
infinite-dimensional Hilbert space, which allows the probe to be initialised in
non-classical states such as squeezed states.
Gaussian states and their metrological properties have been studied extensively in
quantum optics and continuous-variable quantum
information~\cite{Weedbrook2012,Serafini2017,Ferraro2005}, with general expressions
for the QFI in terms of the covariance matrix derived in
Refs.~\cite{Monras2013,Safranek2018,Pinel2013}.
In the thermometry context, Ref.~\cite{Cenni2022} established tight bounds for
Gaussian probe states measured with Gaussian measurements, and
Ref.~\cite{Mirkhalaf2024} showed that non-classical squeezed initial states can
provide a substantial metrological advantage in non-equilibrium Gaussian
thermometry.
Together, these results motivate a unified study of the interplay between
initial-state squeezing, non-Markovian dynamics, and strong coupling.

This is the programme we carry out here.
We study low-temperature non-equilibrium quantum thermometry with a quantum harmonic oscillator
(QHO) probe strongly coupled to a bosonic bath at temperature $T$ through a
Drude--Ohmic spectral density.
Because the probe and bath together form a quadratic system, the dynamics admits an
exact analytical solution, derived by Boyanovsky and
Jasnow~\cite{BoyanoskyJasnow2017}, which forms the backbone of our analysis.
For comparison, we also treat the dynamics within a Markovian description based on
a renormalized Gorini--Kossakowski--Lindblad--Sudarshan (GKLS) master equation for
the probe~\cite{WinczewskiAlicki2021,Lobejko2024}, which accounts for the
coupling-induced renormalisation of the oscillator frequency. Even in this case, this equation cannot provide an accurate description for the dynamics induced in the probe by strong-coupling interactions. The initial probe--bath state is taken to be a product state with the probe in a
general single-mode Gaussian state, including squeezed states, a choice we interpret operationally as a quench preparation protocol.

Our analysis yields four main results.
First, in the non-Markovian regime the QFI behaves non-monotonically in time and
exhibits pronounced revivals, in direct analogy with the fermionic
probe~\cite{RavellRodriguez2024}; the non-Markovian memory of the bath can thus be
harnessed to improve thermometric precision.
Second, in the Markovian regime we prove that the QFI rises monotonically to its
stationary value and develops no interior optimum (Appendix~\ref{app:monotonicity}),
so that the optimal interrogation time always lies at the boundary
$t^*\to\infty$---both for a low-information probe such as the ground state and, as
the closed-form QFI shows, for squeezed preparations---sharpening the picture of
Ref.~\cite{SekatskiPerarnau2022}.
Third, squeezed initial states provide a large transient sensitivity that exceeds that of coherent or thermal states at short times, before thermalization degrades the
squeezing; this identifies squeezing as a genuine thermometric resource and shows
that combining non-Markovian dynamics with squeezed probes shifts the optimal
interrogation time to $t^*>0$, making time itself a resource.
Fourth, in steady-state thermometry strong coupling enhances the QFI relative to
the weak-coupling limit~\cite{Correa2017}: the relative error $T^2/\mathcal{F}$
diverges only polynomially, rather than exponentially, at low temperatures.
Taken together, these results characterise non-equilibrium bosonic thermometry
across the Markovian and non-Markovian regimes, and establish squeezing and
interrogation time as complementary resources for thermometric precision.

The rest of the paper is organised as follows.
Section~\ref{sec:prelim} recalls the necessary background on quantum parameter
estimation, Gaussian states, and Gaussian thermometry.
Section~\ref{sec:setup} introduces the model, the exact solution of
Ref.~\cite{BoyanoskyJasnow2017}, the Markovian master equation, and the
parametrisation of the initial Gaussian probe states.
Our results for steady-state and transient thermometry are presented in
Sec.~\ref{sec:results}.
We discuss the physical implications, the experimental and theoretical outlook,
and draw our conclusions in Sec.~\ref{sec:discussion}.
Technical details of the exact and Markovian solutions, and the proof of QFI
monotonicity, are collected in the appendices.

\section{Preliminaries}
\label{sec:prelim}

\subsection{Quantum parameter estimation}
\label{subsec:qpe}

We assume that the protocol for estimating the temperature of a bosonic sample
proceeds in three stages:
(i) the probe is prepared in a suitable state $\rho(0)$;
(ii) it then interacts with the sample, so that the temperature $T$ becomes
encoded in its state $\rho_T(t)$ after an interaction time $t$;
and (iii) information is extracted by a POVM on the probe, that is, a set of
measurement operators $\{E_x\}_x$ satisfying $E_x\geq 0$ and
$\sum_x E_x = \mathds{1}$~\cite{NielsenChuang2010}.
The outcome probabilities follow Born's rule $p(x|T) = \Tr(E_x \rho_T(t))$, and the
outcomes are mapped to an estimator $\tilde{T}(x)$ that returns an estimate of the
true temperature~\cite{Cramer1999,Rao1992}.

From classical estimation theory, the variance of any locally unbiased
estimator~\cite{Cramer1999,Rao1992} is bounded below by
\begin{align}
  \bigl(\Delta \tilde{T}(x)\bigr)^2 \;\geq\; \frac{1}{I[p(x|T)]},
  \label{eq:cramer-rao}
\end{align}
where $I[p(x|T)]$ is the Fisher information, which depends only on the likelihood
$p(x|T)$ of obtaining outcome $x$ given the true temperature $T$,
\begin{align}
  I[p(x|T)] \;=\; \int p(x|T)
  \left(\frac{\partial}{\partial T}\log p(x|T)\right)^{\!2} dx.
  \label{eq:fisher-classical}
\end{align}
This is the celebrated Cram\'er--Rao bound~\cite{Cramer1999,Rao1992}.
Since it refers to a specific initial state, interaction, and measurement, it is
natural to optimise the measurement so as to maximise the Fisher information.
Optimising over all POVMs yields the quantum Fisher information
(QFI)~\cite{BraunsteinCaves1994,Paris2009}.
Although this optimisation might appear formidable, it has a compact solution:
introducing the symmetric logarithmic derivative (SLD) $L$ through the Lyapunov
equation
\begin{align}
  \frac{\partial \rho_T}{\partial T} \;=\; \frac{L\rho_T + \rho_T L}{2},
  \label{eq:sld}
\end{align}
the QFI is simply $\QFI[\rho_T] = \Tr(L^2 \rho_T)$~\cite{BraunsteinCaves1994}.
The POVM that attains it is the projective measurement in the eigenbasis of $L$,
which may, however, be highly non-local and difficult to
implement~\cite{BraunsteinCaves1994,Paris2009}.
The QFI also admits a geometric interpretation as the rate of change of the state
fidelity,
\begin{align}
  \QFI[\rho_T] \;=\; \lim_{dT\to 0} 8\,\frac{1 - \sqrt{F[\rho_T,\rho_{T+dT}]}}{dT^2},
  \label{eq:qfi-fidelity}
\end{align}
where
$F(\rho_1,\rho_2)=\Tr\!\left(\sqrt{\sqrt{\rho_1}\,\rho_2\sqrt{\rho_1}}\right)^{\!2}$
is the fidelity.
Since $\QFI[\rho_T]\geq I[p(x|T)]$ by construction, one obtains the quantum
Cram\'er--Rao bound
\begin{align}
  \bigl(\Delta \tilde{T}(x)\bigr)^2 \;\geq\; \frac{1}{\QFI[\rho_T]}.
  \label{eq:qcramer-rao}
\end{align}

\subsection{Gaussian states and phase-space formalism}
\label{subsec:gaussian-states}

We now recall the elements of the covariance-matrix formalism for single-mode
continuous-variable systems that are used throughout the paper; comprehensive
treatments may be found in Refs.~\cite{Weedbrook2012,Serafini2017,Ferraro2005}.

For a quantum harmonic oscillator with bare frequency $\Omega$, Hamiltonian
$H_S = p^2/2 + \Omega^2 q^2/2$, and ladder operators $[a,a^\dagger]=1$, the
canonical quadratures are
\begin{align}
  q \;=\; \frac{1}{\sqrt{2\Omega}}\bigl(a+a^\dagger\bigr),
  \qquad
  p \;=\; i\sqrt{\frac{\Omega}{2}}\bigl(a^\dagger - a\bigr),
  \label{eq:quadratures}
\end{align}
satisfying $[q,p]=i$ (units $\hbar=1)$, and we collect them into the vector
$\mathbf{R}=(R_1,R_2)^\top \equiv (q,p)^\top$.
The state of the mode is characterised by its vector of first moments
$\bar{\mathbf{X}} = (\langle q\rangle, \langle p\rangle)^\top$ and its covariance
matrix (CM) $\boldsymbol{\sigma}$, with entries
\begin{align}
  \sigma_{kl}
  \;=\;
  \frac{1}{2}\langle\{R_k,\, R_l\}\rangle - \langle R_k\rangle\langle R_l\rangle.
  \label{eq:def-CM}
\end{align}
The commutation relations $[R_k, R_l] = i J_{kl}$, with
$\sympform=\bigl(\begin{smallmatrix}0&1\\-1&0\end{smallmatrix}\bigr)$ the
symplectic form, impose the Robertson--Schr\"odinger uncertainty principle
\begin{align}
  \boldsymbol{\sigma} + \frac{i}{2}\,\sympform \;\geq\; 0,
  \label{eq:uncertainty}
\end{align}
equivalent for a single mode to $\det\boldsymbol{\sigma}\geq 1/4$, a condition that
every physical state must satisfy~\cite{Serafini2017}.

A state $\rho$ is Gaussian when its Wigner function is a Gaussian distribution on
phase space,
\begin{align}
  W(\mathbf{x})
  \;=\;
  \frac{
    \exp\!\left[-\tfrac{1}{2}
    (\mathbf{x}-\bar{\mathbf{X}})^\top
    \boldsymbol{\sigma}^{-1}
    (\mathbf{x}-\bar{\mathbf{X}})\right]
  }{2\pi\sqrt{\det\boldsymbol{\sigma}}},
  \label{eq:Wigner}
\end{align}
with $\mathbf{x}=(q,p)^\top$ the phase-space point, and is then completely
determined by $\bar{\mathbf{X}}$ and $\boldsymbol{\sigma}$~\cite{Serafini2017}.
Its purity is
\begin{align}
  \mu \;\equiv\; \Tr[\rho^2]
  \;=\;
  \frac{1}{2\sqrt{\det\boldsymbol{\sigma}}},
  \label{eq:purity}
\end{align}
so that $\mu=1$ exactly when $\det\boldsymbol{\sigma}=1/4$, i.e.\ when the state is
pure and saturates Eq.~\eqref{eq:uncertainty}.

The states relevant to thermometry form a short hierarchy.
The vacuum $|0\rangle$ has $\bar{\mathbf{X}}=\mathbf{0}$ and
\begin{align}
  \boldsymbol{\sigma}_\mathrm{vac}
  \;=\;
  \frac{1}{2}
  \begin{pmatrix} 1/\Omega & 0 \\ 0 & \Omega \end{pmatrix},
  \label{eq:CM-vac}
\end{align}
while the thermal Gibbs state $\pi_T = e^{-H_S/T}/Z$ likewise has
$\bar{\mathbf{X}}=\mathbf{0}$ but a temperature-dependent CM,
\begin{align}
  \boldsymbol{\sigma}_\mathrm{th}
  \;=\;
  \left(\bar{n}+\tfrac{1}{2}\right)
  \begin{pmatrix} 1/\Omega & 0 \\ 0 & \Omega \end{pmatrix}
  \;=\;
  \frac{\coth(\Omega/2T)}{2}
  \begin{pmatrix} 1/\Omega & 0 \\ 0 & \Omega \end{pmatrix},
  \label{eq:CM-thermal}
\end{align}
where $\bar{n}(\Omega,T) = (e^{\Omega/T}-1)^{-1}$ is the Bose--Einstein occupation.
Because $\boldsymbol{\sigma}_\mathrm{th}$ depends on $T$ only through $\bar{n}$, it
is the natural object to differentiate with respect to $T$ when computing the QFI.
Displacing the vacuum by $D(\alpha)=\exp(\alpha a^\dagger - \alpha^* a)$, with
$\alpha=|\alpha|e^{i\theta}$, shifts the first moments to
$\bar{\mathbf{X}} =
(\sqrt{2/\Omega}\,|\alpha|\cos\theta,\,\sqrt{2\Omega}\,|\alpha|\sin\theta)^\top$
while leaving the CM unchanged; the resulting coherent states
$|\alpha\rangle = D(\alpha)|0\rangle$ saturate the uncertainty principle and carry
equal noise in both quadratures, up to the $\Omega$-dependent rescaling.

Squeezing introduces the non-classical features central to this work.
The squeezing operator
$S(\xi)=\exp\!\left(\frac{\xi^* a^2 - \xi a^{\dagger 2}}{2}\right)$, with
$\xi = r\,e^{2i\varphi}$ $(r\geq 0$, $\varphi\in[0,\pi))$, produces the squeezed
vacuum $S(\xi)|0\rangle$, which has $\bar{\mathbf{X}}=\mathbf{0}$ and
\begin{align}
  \boldsymbol{\sigma}_\mathrm{sq}
  \;=\;
  \frac{1}{2}
  \begin{pmatrix}
    \dfrac{\cosh 2r - \cos 2\varphi\,\sinh 2r}{\Omega}
    &
    \sin 2\varphi\,\sinh 2r
    \\[8pt]
    \sin 2\varphi\,\sinh 2r
    &
    \Omega\!\left(\cosh 2r + \cos 2\varphi\,\sinh 2r\right)
  \end{pmatrix}.
  \label{eq:CM-squeezed}
\end{align}
The parameter $r$ sets the degree of noise reduction: for $\varphi=0$, squeezing
along $q$, the position variance is reduced to $\sigma_{11}=e^{-2r}/(2\Omega)$,
a factor $e^{-2r}$ below the vacuum level, while the momentum variance is amplified
to $\sigma_{22}=\Omega e^{2r}/2$; as $r\to 0$ one recovers
Eq.~\eqref{eq:CM-vac}, and the angle $\varphi$ fixes the orientation of the
squeezed quadrature, with $\varphi=0$ optimal for thermometry through position
measurement.
The most general single-mode pure Gaussian state is the displaced squeezed state
$|\alpha,\xi\rangle = D(\alpha)\,S(\xi)|0\rangle$, with CM given by
Eq.~\eqref{eq:CM-squeezed} and first moments as in the coherent state.
Its mean energy,
\begin{align}
  E
  \;=\;
  \langle H_S\rangle
  \;=\;
  \Omega\!\left(|\alpha|^2 + \frac{\cosh 2r}{2}\right),
  \label{eq:energy-general}
\end{align}
separates into the displacement energy $\Omega|\alpha|^2$ and the squeezing energy
$\Omega\cosh(2r)/2$, reducing to the zero-point energy $\Omega/2$ when
$r=|\alpha|=0$.
This decomposition underlies the energy-constrained comparison of probe states
introduced in Sec.~\ref{subsec:param}.

\subsection{Gaussian thermometry}
\label{subsec:gaussian}

In this section, we describe the theoretical results in thermometry for Gaussian
states and measurements.
These results are by no means new, and we closely follow the exposition of
Ref.~\cite{Cenni2022}.

Gaussian measurements are characterised by a displacement vector $\mathbf{d}^M_s$
and a covariance matrix $\boldsymbol{\sigma}^M_s$, where $s$ denotes the POVM
setting.
For a single-mode probe $(m=1)$, the probability of obtaining an outcome
$\mathbf{a}$ is
\begin{align}
  p(\mathbf{a}|\rho,s)
  \;=\;
  \frac{
    \exp\!\left[{-\tfrac{1}{2}
    \bigl(\mathbf{d}^M_s+\mathbf{a}-\mathbf{d}\bigr)^{\!\top}
    \bigl(\boldsymbol{\sigma}^M_s+\boldsymbol{\sigma}\bigr)^{-1}
    \bigl(\mathbf{d}^M_s+\mathbf{a}-\mathbf{d}\bigr)}\right]
  }{
    2\pi \sqrt{\det(\boldsymbol{\sigma}^M_s+\boldsymbol{\sigma})}
  },
  \label{eq:gaussian-prob}
\end{align}
where $\mathbf{d}$ and $\boldsymbol{\sigma}$ are the displacement vector and CM of
the probe state, respectively.
Since this distribution is Gaussian, the classical Fisher information can be
computed analytically~\cite{Cenni2022}:
\begin{equation}
  I[p(\mathbf{a}|\rho,s)]
  = \partial_T \mathbf{d}^{\!\top}
     \bigl(\boldsymbol{\sigma}+\boldsymbol{\sigma}^M_s\bigr)^{-1}
     \partial_T \mathbf{d}
  +
     \frac{1}{2}\Tr\!\left[
       \Bigl(\bigl(\boldsymbol{\sigma}+\boldsymbol{\sigma}^M_s\bigr)^{-1}
       \partial_T\boldsymbol{\sigma}\Bigr)^{\!2}
     \right],
  \label{eq:cfi-gaussian}
\end{equation}
where we have omitted the dependence on $T$ for conciseness.
Note that the second term involves the \emph{trace of the square} of the matrix
$(\boldsymbol{\sigma}+\boldsymbol{\sigma}^M_s)^{-1}\partial_T\boldsymbol{\sigma}$. 

In this work, we focus on the homo- and heterodyne measurements as well as the optimal measurement, i.e., the eigenbasis of the SLD. For homodyne measurement of the position quadrature,
$\boldsymbol{\sigma}^M_s = \lim_{r\to\infty}\operatorname{diag}(1/r,\,r)$~\cite{Cenni2022},
and Eq.~\eqref{eq:cfi-gaussian} reduces to
\begin{align}
  I[p(\mathbf{a}|\rho,\mathrm{homo})]
  \;=\;
  \frac{(\partial_T \sigma_{11})^2}{2\sigma_{11}^2}
  \;+\;
  \frac{(\partial_T d_1)^2}{\sigma_{11}}.
  \label{eq:cfi-homo}
\end{align}
Heterodyne detection corresponds to a simultaneous measurement of both quadratures
via projection onto coherent states. The covariance matrix of the measurement is the identity matrix, $\boldsymbol{\sigma}^M_s=\mathbb{I}$, so, 
equation~\eqref{eq:cfi-gaussian} then yields
\begin{align}
  I[p(\mathbf{a}|\rho,\mathrm{hetero})]
  \;=\;
  \partial_T \mathbf{d}^{\!\top}
  \bigl(\boldsymbol{\sigma}+\mathbb{I}\bigr)^{-1}
  \partial_T \mathbf{d}
  \;+\;
  \frac{1}{2}
  \Tr\!\left[\Bigl(
    \bigl(\boldsymbol{\sigma}+\mathbb{I}\bigr)^{-1}
    \partial_T\boldsymbol{\sigma}
  \Bigr)^{\!2}\right].
  \label{eq:cfi-hetero}
\end{align}
For zero-mean probe states $(\mathbf{d}=\mathbf{0})$ the first term vanishes and
heterodyne performance is determined solely by the temperature sensitivity of the
CM. In our case, though, even if $\mathbf{d} \neq 0$, the derivative of the vector with respect to the temperature is zero. Finally, the QFI also has a closed form for Gaussian states, originally derived in
Refs.~\cite{Monras2013,Pinel2013}.
Decomposing the SLD in quadrature moments,
\begin{align}
  L \;=\; \sum_{kl} C^{(2)}_{kl} R_k R_l
        \;+\; \sum_k C^{(1)}_k R_k
        \;+\; C^{(0)},
  \label{eq:sld-gaussian}
\end{align}
where $C^{(2)}$ is a real $2\times 2$ matrix, $C^{(1)}$ a real $2$-vector, and
$C^{(0)}\in\mathbb{R}$, the coefficients are determined by
\begin{align}
  \partial_T \boldsymbol{\sigma}
    &= 2\boldsymbol{\sigma} C^{(2)} \boldsymbol{\sigma}
     + \frac{1}{2}\sympform C^{(2)} \sympform,
  \label{eq:C2-eq}\\
  C^{(1)}
    &= \boldsymbol{\sigma}^{-1}\partial_T\mathbf{d}
     - 2C^{(2)}\mathbf{d},
  \label{eq:C1-eq}\\
  C^{(0)}
    &= -\bigl(C^{(1)}\bigr)^{\!\top}\mathbf{d}
     - \Tr\bigl[C^{(2)}\boldsymbol{\sigma}\bigr]
     - \mathbf{d}^{\!\top} C^{(2)}\mathbf{d},
  \label{eq:C0-eq}
\end{align}
and the QFI reads
\begin{align}
  \QFI(\mathbf{d},\boldsymbol{\sigma})
  \;=\;
  \partial_T\mathbf{d}^{\!\top}\boldsymbol{\sigma}^{-1}\partial_T\mathbf{d}
  \;+\;
  2\Tr\!\left[
    \Bigl(C^{(2)}\boldsymbol{\sigma}\Bigr)^{\!2}
    + \Bigl(\frac{1}{2}C^{(2)}\sympform\Bigr)^{\!2}
  \right].
  \label{eq:qfi-gaussian}
\end{align}
\section{Setup}
\label{sec:setup}

\subsection{Model: bosonic probe and bath}
\label{subsec:model}

The physical system of interest is a quantum harmonic oscillator (QHO) of unit mass
interacting with a bosonic reservoir, the latter being itself an infinite
collection of non-interacting oscillators.
In the metrological scenario the QHO is the probe, and the reservoir is the sample
whose temperature $T$ we wish to estimate.
The Hamiltonian of the joint system is\footnote{In this section we closely follow
the notation of Ref.~\cite{BoyanoskyJasnow2017}.}
\begin{align}
  H \;=\; H_S + H_B + H_{SB},
  \label{eq:H-total}
\end{align}
with
\begin{align}
  H_S  &= \frac{p^2}{2} + \frac{\osc^2}{2}\,q^2, \label{eq:H-S}\\
  H_B  &= \sum_k \frac{1}{2}\!\left(P_k^2 + W_k^2\,Q_k^2\right), \label{eq:H-B}\\
  H_{SB} &= -q\sum_k C_k\,Q_k. \label{eq:H-SB}
\end{align}
Here $H_S$ is the free Hamiltonian of the probe, with canonical pair $(q,p)$ and
bare frequency $\Omega$, and $H_B$ is the free Hamiltonian of the bath, described
by canonical pairs $\{(Q_k,P_k)\}$ with frequencies $\{W_k\}$.
The coupling~\eqref{eq:H-SB} is position-like: the probe coordinate $q$ couples
bilinearly to each bath coordinate $Q_k$ through the real charges $\{C_k\}$.
In the continuum limit, the bath is fully characterised by its spectral density,
\begin{align}
  J(\omega)
  \;=\; \sum_k \frac{\pi C_k^2}{2W_k}
  \bigl[\delta(\omega-W_k)-\delta(\omega+W_k)\bigr],
  \label{eq:spectral-density-def}
\end{align}
so that all physical quantities depend on the discrete charges $\{C_k\}$ only
through $J(\omega)$; the charges themselves carry no independent physical
meaning.

We specialise to the Drude--Ohmic spectral
density~\cite{BoyanoskyJasnow2017,CaldeiraLeggett1983}
\begin{align}
  J(\omega)
  \;=\; \frac{\gamma\,\omega\,\Lambda^2}{\omega^2+\Lambda^2},
  \label{eq:spectral-Drude}
\end{align}
where $\gamma>0$ is the dimensionless coupling strength and $\Lambda$ the
high-frequency cut-off.
This density interpolates between two familiar limits: at low frequencies
$\omega\ll\Lambda$ it reduces to the Ohmic form $J(\omega)\approx\gamma\omega$,
whereas for $\omega\gg\Lambda$ it decays as $\gamma\Lambda^2/\omega$, supplying a
physically motivated ultraviolet regularisation.
The cut-off thus sets the timescale $\Lambda^{-1}$ below which the bath retains
memory and non-Markovian effects become relevant; in the limit $\Lambda\to\infty$
at fixed $\gamma$ the density becomes strictly Ohmic and the bath memoryless,
recovering the Markovian regime.
Throughout this work we set $\Lambda=100$ and $\Omega=1$ in natural units, placing
the dynamics in the high-cutoff regime of Ref.~\cite{BoyanoskyJasnow2017}, and we
vary $\gamma$ over several orders of magnitude (underdamped regime), subject to the stability condition
derived below.

Finally, the initial probe--bath state is taken to be a product state,
\begin{align}
  \rho(0) \;=\; \rho_S(0)\otimes\pi_B,
  \label{eq:initial-state}
\end{align}
where $\pi_B = e^{-H_B/T}/Z_B$ is the thermal state of the bath at temperature $T$
and $\rho_S(0)$ is a general single-mode Gaussian state of the probe. The product-state assumption is standard when dealing with open-system evolutions. 

\subsection{Exact solution}
\label{subsec:exact}

Because the total Hamiltonian~\eqref{eq:H-total} is quadratic in the canonical
variables, the Heisenberg equations of motion are linear and can be solved exactly.
Eliminating the bath, one obtains a quantum Langevin equation for the probe
coordinate~\cite{BoyanoskyJasnow2017},
\begin{align}
  \ddot{q}(t) + \osc^2\,q(t)
  \;+\; \int_0^t ds\; \Sigma(t-s)\,q(s)
  \;=\; \xi(t),
  \label{eq:langevin}
\end{align}
in which the dissipation kernel follows from the spectral
density~\eqref{eq:spectral-Drude} via
$\Sigma(\tau)=-(2/\pi)\int_0^\infty d\omega\,J(\omega)\sin(\omega\tau)$, which for the
Drude--Ohmic form evaluates to
$\Sigma(\tau)=-\gamma\Lambda^2\,e^{-\Lambda|\tau|}\,\mathrm{sign}(\tau)$, so that the
memory term in Eq.~\eqref{eq:langevin} reads
$-\gamma\Lambda^2\int_0^t ds\,e^{-\Lambda(t-s)}q(s)$~\cite{BoyanoskyJasnow2017}, and
\begin{align}
  \xi(t)
  \;=\; \sum_k C_k
  \!\left[Q_k(0)\cos(W_k t)
          + \frac{P_k(0)}{W_k}\sin(W_k t)\right]
  \label{eq:noise}
\end{align}
is the quantum noise operator carrying the bath initial conditions.
The two-time correlation functions of $\xi(t)$, fully determined by $J(\omega)$ and
$T$, are given in Appendix~\ref{app:exact}.

The dissipation kernel renormalises the oscillator frequency.
In the high-cutoff regime its static $(\omega\to0)$ limit shifts $\osc^2$ by
$-\gamma\Lambda$, since $\gamma\Lambda^2\!\int_0^t e^{-\Lambda(t-s)}ds\to\gamma\Lambda$
for $t\gg\Lambda^{-1}$, so the long-time dynamics is governed by the renormalized
frequency
\begin{align}
  \oscR^2 \;\equiv\; \osc^2 - \gamma\Lambda,
  \label{eq:OmegaR}
\end{align}
which remains real, and the probe dynamically stable, provided
\begin{align}
  \gamma \;<\; \frac{\osc^2}{\Lambda}.
  \label{eq:stability}
\end{align}
This reproduces the renormalisation and stability condition of
Ref.~\cite{BoyanoskyJasnow2017} [their Eqs.~(II.33),(II.35)]; the Laplace transform
of the kernel, $\tilde\Sigma(s)=-\gamma\Lambda^2/(\Lambda+s)$, supplies both the
damping $\gamma s$ and the shift $-\gamma\Lambda$ at large $\Lambda$.
When this bound is violated the effective potential turns imaginary, so we restrict
throughout to the stable regime~\eqref{eq:stability}.
The character of the relaxation is then fixed by the effective oscillation
frequency
\begin{align}
  \effW \;\equiv\; \sqrt{\oscR^2 - \frac{\gamma^2}{4}}.
  \label{eq:W}
\end{align}
For $\effW$ real, i.e.\ $\oscR>\gamma/2$, the probe is under-damped: it performs
damped oscillations at frequency $\effW$, and these oscillations are the origin of
the non-Markovian memory effects and the non-monotonic QFI found below.
For $\effW$ imaginary, i.e.\ $\oscR<\gamma/2$, the probe is over-damped and relaxes
monotonically to its steady state.
With our parameters $\Lambda=100$ and $\Omega=1$, the under-damped condition
$\oscR>\gamma/2$ fails only in a vanishingly thin window adjacent to the stability
boundary $\gamma=\Omega^2/\Lambda=0.01$, so all couplings considered here lie in the
under-damped regime.
The complete analytical solution for the covariance matrix
$\boldsymbol{\sigma}(t)$ is collected in Appendix~\ref{app:exact}. The frequency integrals determining the exact time-dependent covariance
matrix are evaluated numerically using custom GPU-parallel CUDA
routines~\cite{WinczewskiCUDA2024}.

\subsection{Markovian approximation}
\label{subsec:markov}

For comparison with the exact dynamics, we also describe the probe within a
Markovian, renormalized GKLS master
equation~\cite{WinczewskiAlicki2021,Lobejko2024}.
The essential point~\cite{WinczewskiAlicki2021,Lobejko2024} is to absorb the
coupling-induced frequency renormalisation by replacing the bare frequency
$\Omega$ with $\Omega_R$ throughout.
Defining ladder operators at the renormalized frequency,
\begin{align}
  a \;=\; \sqrt{\frac{\oscR}{2}}\,q + \frac{i}{\sqrt{2\oscR}}\,p,
  \qquad
  a^\dagger \;=\; \sqrt{\frac{\oscR}{2}}\,q - \frac{i}{\sqrt{2\oscR}}\,p,
  \label{eq:ladder}
\end{align}
which differ from those in Eq.~\eqref{eq:quadratures} associated with the bare
frequency, the master equation reads
\begin{align}
  \partial_t\rho
  \;=\; -i\bigl[\oscR\,a^\dagger a,\,\rho\bigr]
  \;+\; \gamma_\downarrow\,\diss{a}\rho
  \;+\; \gamma_\uparrow\,\diss{a^\dagger}\rho,
  \label{eq:GKLS}
\end{align}
with $\diss{A}\rho = A\rho A^\dagger - \frac{1}{2}\{A^\dagger A,\rho\}$ and rates
\begin{align}
  \gamma_\downarrow &\;=\; \gamma\bigl[\bar{n}(\oscR)+1\bigr], &
  \gamma_\uparrow   &\;=\; \gamma\,\bar{n}(\oscR),
  \label{eq:rates}
\end{align}
where $\bar{n}(\oscR)=(e^{\oscR/T}-1)^{-1}$.
The net decay rate $\gamma_\downarrow-\gamma_\uparrow=\gamma$ is independent of
$T$, and the unique steady state is the Gibbs state at the renormalized frequency,
which correctly captures the energy renormalisation induced by strong coupling.
The corresponding analytical solution for $\boldsymbol{\sigma}(t)$ is given in
Appendix~\ref{app:markov}.

\subsection{Initial Gaussian probe states}
\label{subsec:param}

The probe is initialised in a general single-mode pure Gaussian state, the
displaced squeezed state $D(\alpha)S(\xi)|0\rangle$ of
Sec.~\ref{subsec:gaussian-states}, now built from the ladder
operators~\eqref{eq:ladder} at the renormalized frequency $\oscR$.
Parametrising it by the displacement $\alpha=|\alpha|e^{i\theta}$ and the squeezing
$\xi=re^{2i\varphi}$, its covariance matrix is Eq.~\eqref{eq:CM-squeezed} with
$\Omega\to\oscR$,
\begin{align}
  \boldsymbol{\sigma}(0)
  \;=\;
  \frac{1}{2}
  \begin{pmatrix}
    \dfrac{\cosh 2r - \cos2\varphi\,\sinh 2r}{\oscR}
    & \sin2\varphi\,\sinh 2r \\[8pt]
    \sin2\varphi\,\sinh 2r
    & \oscR\!\left(\cosh 2r + \cos2\varphi\,\sinh 2r\right)
  \end{pmatrix},
  \label{eq:CM-pure}
\end{align}
and its mean energy is
\begin{align}
  E
  \;=\; \langle H_S \rangle
  \;=\; \oscR\!\left(|\alpha|^2 + \frac{\cosh 2r}{2}\right).
  \label{eq:energy}
\end{align}

To compare probe states fairly we hold their energy fixed, following
Ref.~\cite{RavellRodriguezMorelli2024}, and redistribute it between the
displacement energy $\oscR|\alpha|^2$ and the squeezing energy
$\oscR\cosh(2r)/2$.
At fixed $E$ and squeezing angle $\varphi$, Eq.~\eqref{eq:energy} ties the two
parameters together through
\begin{align}
  |\alpha|^2 \;=\; \frac{E}{\oscR} - \frac{\cosh 2r}{2},
  \label{eq:energy-constraint}
\end{align}
so that increasing the squeezing $r$ necessarily reduces the displacement
$|\alpha|$.
The two extremes are the purely squeezed state, $|\alpha|=0$ with
$r=\frac{1}{2}\operatorname{arccosh}(2E/\oscR)$, and the coherent state, $r=0$ with
$|\alpha|^2=E/\oscR-1/2$.
Squeezing along $q$ $(\varphi=0)$ minimises $\sigma_{11}$ and thereby sharpens the
probe's sensitivity to bath-induced changes in $\langle q^2\rangle$; as we show in
Sec.~\ref{sec:results}, this produces a much larger transient QFI than a coherent
or thermal state of the same energy, until the squeezing is eventually erased by
thermalisation.

\section{Results}
\label{sec:results}

We organise our results around the two regimes of interest.
Section~\ref{subsec:steady} treats the steady state, reached once the probe has
fully thermalised with the bath, while Sec.~\ref{subsec:transient} turns to the
transient regime, where the probe is measured before thermalisation; the latter
discussion proceeds from the non-Markovian exact dynamics to its Markovian
counterpart, then to the role of initial squeezing, and finally to the dependence
on coupling strength.
Throughout this section we set $\hbar = k_B = 1$, $\Omega = 1$, and $\Lambda = 100$, vary the
coupling $\gamma$ subject to the stability condition~\eqref{eq:stability}.

\subsection{Steady-state thermometry}
\label{subsec:steady}

We first consider the probe fully thermalised with the bath.
A natural figure of merit is the dimensionless relative error bound
\begin{align}
  \left(\frac{\delta T}{T}\right)^2
  \;\geq\;
  \frac{1}{T^2\,\QFI[\rho_T^\mathrm{ss}]},
  \label{eq:rel-error}
\end{align}
where $\rho_T^\mathrm{ss}$ is the steady state of the probe at temperature $T$.
Figure~\ref{fig:shot-noise} shows $1/(T^2 \QFI)$ as a function of the reduced
temperature $T/\oscR$ for several coupling strengths, comparing the exact solution
(solid coloured lines) with the prediction of the renormalized GKLS
equation~\eqref{eq:GKLS} (dashed).

\begin{figure}[h]
  \centering
  \includegraphics[width=0.7\linewidth]{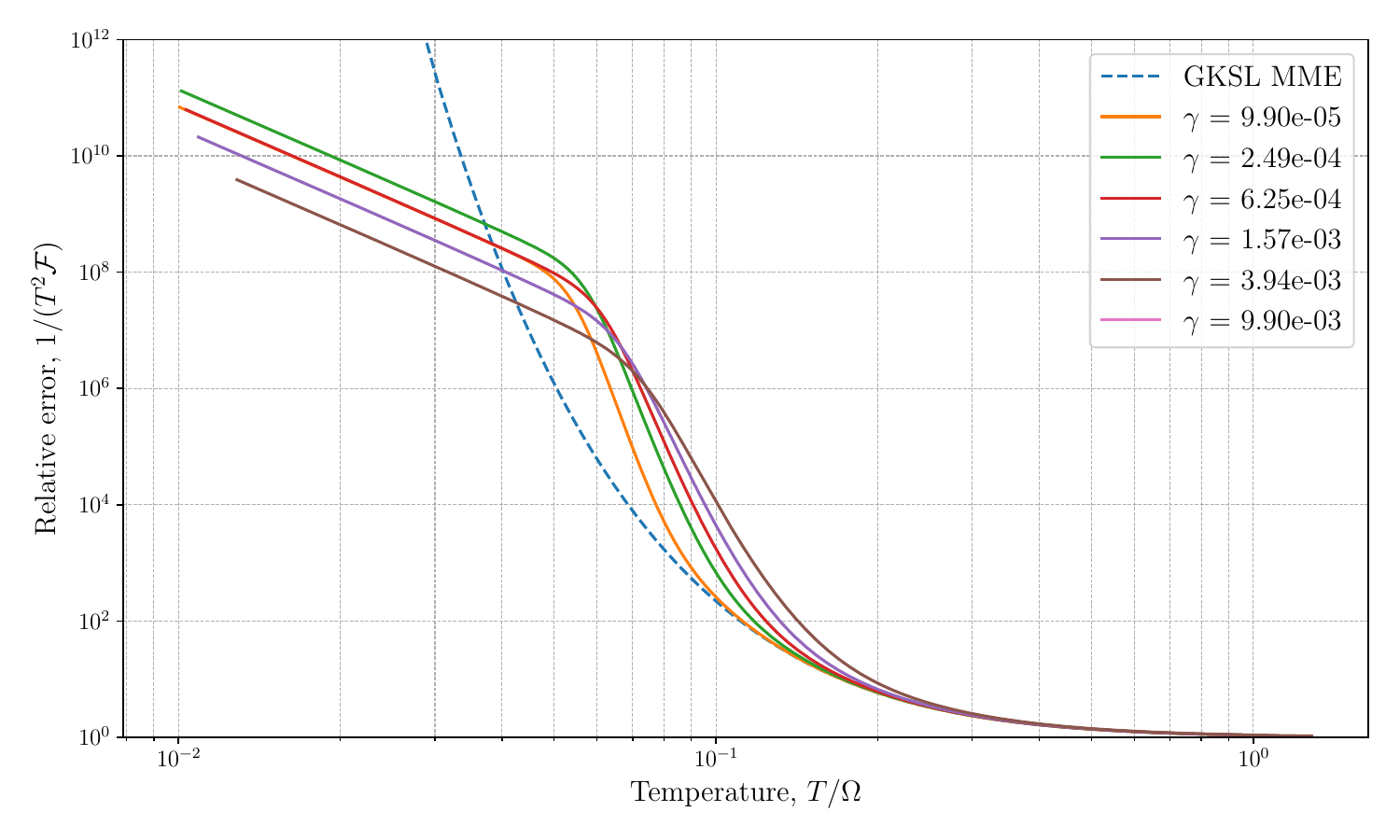}
  \caption{%
    Log-log plot of the minimum relative error $1/(T^2\,\QFI)$ at steady state
    as a function of $T/\osc$, for
    $\gamma\in\{9.90\times10^{-5},\;2.49\times10^{-4},\;6.25\times10^{-4},\;
    1.57\times10^{-3},\;3.94\times10^{-3},\;9.90\times10^{-3}\}$
    (solid coloured lines).
    The dashed line is the Markovian (GKLS) prediction, which is independent of
    $\gamma$.
    Parameters: $\Omega=1$, $\Lambda=100$.
  }
  \label{fig:shot-noise}
\end{figure}

The figure reveals a clear three-part structure.
At high temperatures $T\gg\oscR$ all curves collapse onto one another, recovering
the classical equipartition regime in which the relative error is essentially
coupling-independent. As the temperature is lowered the curves separate, and the contrast between the
Markovian and exact descriptions becomes the dominant feature. As expected, for the smaller values of $\gamma$, the Markovian prediction matches the exact solution in a wider window of temperature values. The Markovian approximation gives a lower relative error than the exact calculations, for some values of $T/\Omega_R$ but these correspond to the regime in which the Markovian master equation has a much limited accuracy, rendering the Markovian prediction unreliable in that regime. 

The Markovian prediction can be obtained in closed form, which makes the contrast
quantitative.
The steady state of Eq.~\eqref{eq:GKLS} is the Gibbs state at the renormalized
frequency $\oscR$, for which the temperature dependence resides entirely in the
energy distribution.
Evaluating the QFI for this thermal state---equivalently, recognising that the
optimal measurement is energy measurement and that the Fisher information of the
resulting Boltzmann distribution is the variance of the energy divided by
$T^4$---gives the standard thermal result~\cite{Correa2015,Mehboudi2019}
\begin{align}
  \QFI^\mathrm{Markov}_\mathrm{ss}
  \;=\;
  \frac{\oscR^2}{T^4}\,\bar{n}(\bar{n}+1),
  \qquad
  \bar{n} = \frac{1}{e^{\oscR/T}-1}.
  \label{eq:markov-qfi}
\end{align}
This expression is independent of $\gamma$, which is why the Markovian curve in
Fig.~\ref{fig:shot-noise} is universal.
At low temperature $T\ll\oscR$ one has $\bar{n}\approx e^{-\oscR/T}$, and
Eq.~\eqref{eq:markov-qfi} yields an exponentially diverging relative error,
\begin{align}
  \frac{1}{T^2\,\QFI^\mathrm{Markov}_\mathrm{ss}}
  \;=\;
  \frac{T^2}{\oscR^2\,\bar{n}(\bar{n}+1)}
  \;\approx\;
  \frac{T^2}{\oscR^2}\,e^{\,\oscR/T},
  \qquad T \ll \oscR,
  \label{eq:markov-scaling}
\end{align}
the familiar Boltzmann suppression of thermometric precision at ultralow
temperatures~\cite{Correa2015,HovhannisyanCorrea2018}.

The exact solution departs sharply from this prediction as the coupling grows.
At the strongest coupling shown $(\gamma\approx9.9\times10^{-3}$, near the
stability boundary), the relative error diverges only polynomially as the
temperature is lowered,
\begin{align}
  \frac{1}{T^2\,\QFI^\mathrm{exact}}
  \;\sim\;
  \left(\frac{T}{\oscR}\right)^{-2-\eta},
  \qquad T \ll \oscR,
  \label{eq:exact-scaling}
\end{align}
with the exponent $\eta \approx 2.03$ read off from the slope of Fig.~\ref{fig:shot-noise}, which is roughly the same value as the one obtained in the fermionic setting~\cite{RavellRodriguez2024}.

This exponent is not accidental. At low temperature the thermal part of the
steady-state position variance scales as
$\sigma_{11}(T)-\sigma_{11}(0)\sim T^{s+1}$ for a bath whose low-frequency
spectral density behaves as $J(\omega)\sim\omega^{s}$: by the
fluctuation--dissipation theorem the response satisfies
$\operatorname{Im}\chi(\omega)\sim\omega^{s}$ at small $\omega$, and
$\int_0^\infty d\omega\,\omega^{s}/(e^{\omega/T}-1)\sim T^{s+1}$.
The temperature sensitivity of the variance is then $\partial_T\sigma_{11}\sim
T^{s}$, the position-quadrature QFI scales as $(\partial_T\sigma_{11})^2/\sigma_{11}^2
\sim T^{2s}$, and the relative error as $T^{-2-2s}$, identifying $\eta=2s$.
For the Ohmic case $s=1$ realised by the Drude--Ohmic density~\eqref{eq:spectral-Drude}
this predicts $\eta=2$, consistent with the measured $\eta\approx2.03$.

The replacement of the exponential scaling~\eqref{eq:markov-scaling} by a
power law is a qualitative, not merely quantitative, improvement, and it follows
from the low-frequency behaviour $J(\omega)\sim\gamma\omega$ of the Drude--Ohmic
bath: at strong coupling the probe inherits the bath's low-frequency thermal
fluctuations, circumventing the Boltzmann suppression that constrains the weakly
coupled (Markovian) probe.
In this sense our results extend the equilibrium strong-coupling enhancement of
Ref.~\cite{Correa2017} to the Drude--Ohmic bath with finite cut-off $\Lambda$. This polynomial scaling has also been observed in other models~\cite{HovhannisyanCorrea2018,Potts2019,Jorgensen_2020,Mihailescu_2023}.

\subsection{Transient thermometry}
\label{subsec:transient}

We now measure the probe at a finite interaction time $t$, before thermalisation,
and follow the QFI together with the homodyne and heterodyne Fisher informations
as the physical situation is varied.
A point worth stating at the outset is that the QFI vanishes at $t=0$ for every
preparation: at the instant of contact the probe is still in its
temperature-independent initial state and uncorrelated with the bath, so
$\partial_T\rho_S(0)=0$ and $\QFI(0)=0$.
All temperature information is acquired dynamically, and the transient ``advantages''
discussed below are therefore statements about how, and how fast, the QFI builds up
from zero.

Figure~\ref{fig:ex-diff-temp} shows these three quantities as functions of time for
the exact solution at strong coupling $(\gamma=9.901\times10^{-3}$, giving
$\oscR\approx0.0995)$ and three temperatures $T\in\{10^{-3},10^{-2},10^{-1}\}$.

\begin{figure}[h]
  \centering
  \includegraphics[width=0.5\linewidth]{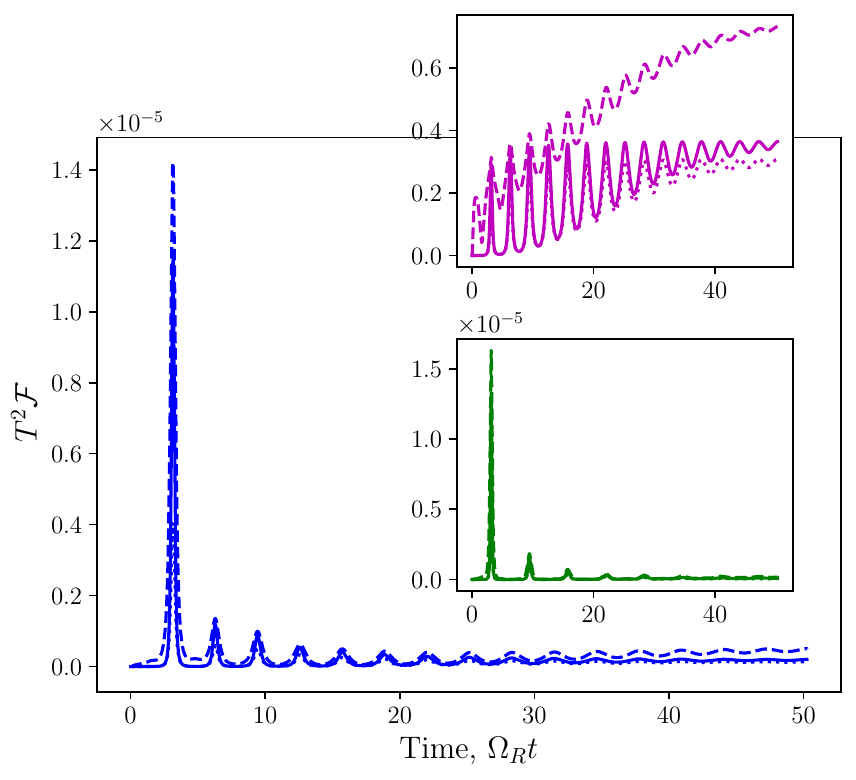}
  \caption{%
    Time dependence of the metrological sensitivity for the exact solution times temperature squared.
    Dashed: QFI; solid: homodyne CFI; dotted: heterodyne CFI.
    Colours: $T=10^{-3}$ (blue, main plot), $T=10^{-2}$ (green, bottom inset), $T=10^{-1}$ (magenta, top inset).
    Parameters: $\gamma=9.901\times10^{-3}$, $\Lambda=100$, $\Omega=1$, and $\Omega_R \approx 0.99$.
    Initial state: ground state $(r=|\alpha|=0)$.
  }
  \label{fig:ex-diff-temp}
\end{figure}

The defining feature is that for the exact solution the QFI---and also its classical version associated with homo- and heterodyne measurements---is non-monotonic in time: rather than rising
smoothly to its steady-state value, it undergoes revivals on the timescale set by
the effective oscillation period $2\pi/\effW$.
These revivals are the direct signature of the bath memory encoded in the
kernel~\eqref{eq:langevin}: information that has flowed from the probe into the
bath is partially returned at each oscillation cycle, transiently lifting the
sensitivity above its long-time value.
The same phenomenon was found for the fermionic probe in
Ref.~\cite{RavellRodriguez2024}, which confirms that a non-monotonic transient QFI
is generic to strongly coupled thermometry with quadratic Hamiltonians, regardless
of statistics.
The effect is most pronounced at the lowest temperature, where the peak transient
QFI clearly exceeds the steady-state value, so that a finite interrogation time
$t^*>0$ strictly outperforms waiting for thermalisation.
As for the measurement, the homodyne curve closely follows the QFI at all times
and temperatures, so that a simple position-quadrature measurement is essentially
optimal and the highly non-local SLD measurement is unnecessary; the heterodyne
curve lies consistently below, confirming that the temperature information is
carried mainly by the position variance $\sigma_{11}(t)$, as expected for a bath
whose spectral density grows linearly at low frequencies.

Replacing the exact dynamics by the renormalized GKLS
equation~\eqref{eq:GKLS} changes the picture qualitatively, as
Fig.~\ref{fig:Mark-diff-temp} shows.

\begin{figure}[h]
  \centering
  \includegraphics[width=0.5\linewidth]{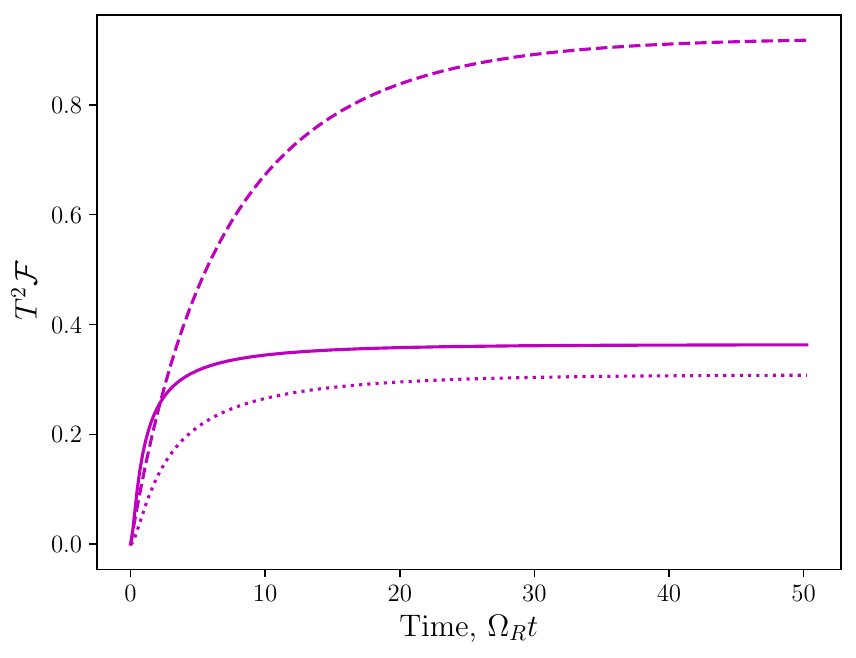}
  \caption{%
    Time dependence of the metrological sensitivity times temperature squared in the Markovian approximation.
    Parameters: $\gamma=9.901\times10^{-3}$, $\Lambda=100$, $\Omega=1$, $\Omega_R \approx 0.99$, and $T=10^{-1}$.
    Initial state: ground state $(r=|\alpha|=0)$. 
  }
  \label{fig:Mark-diff-temp}
\end{figure}

In contrast to the exact dynamics, the Markovian QFI grows monotonically towards
its steady-state value, as established in Appendix~\ref{app:monotonicity}; for the
ground-state preparation shown here it approaches the steady state from below
(Proposition~\ref{prop:iso}).
The optimal interrogation time is therefore pushed to $t^*\to\infty$: the best one
can do is wait for full thermalisation, in agreement with and reinforcing
Ref.~\cite{SekatskiPerarnau2022}.
The two descriptions also differ enormously in magnitude.
At $T=10^{-3}$ the Markovian steady-state QFI is, from
Eq.~\eqref{eq:markov-qfi}, $\QFI^\mathrm{Markov}_\infty\approx10^{-33}$, whereas the
exact QFI peaks at order $10^{-1}$ (Fig.~\ref{fig:ex-diff-temp})---a suppression of
some thirty orders of magnitude.
The non-Markovian bath memory, retained by the exact solution but discarded by the
master equation, is thus the physical ingredient that makes precise thermometry at
ultralow temperatures possible at all.

We now turn from the dynamics to the initial state, asking how squeezing affects
the transient QFI.
Figure~\ref{fig:diff-r} shows the QFI, together with the homodyne and heterodyne
Fisher informations in the insets, for an initial squeezed vacuum $S(r)|0\rangle$
with $r\in\{0,1,3,5\}$ and squeezing angle $\varphi=0$, at $T=10^{-3}$.
The states are compared at fixed total energy $E$: as $r$ grows, energy is
transferred from displacement to squeezing according to the constraint given in
Eq.~\eqref{eq:energy-constraint}.

\begin{figure}[h]
  \centering
  \includegraphics[width=0.5\linewidth]{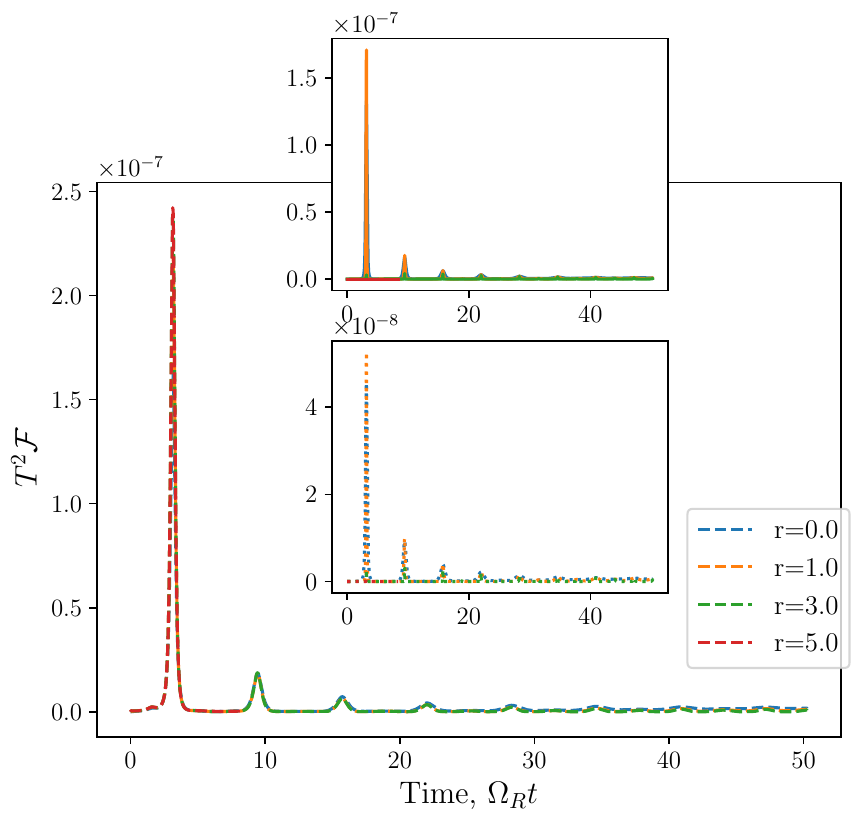}
  \caption{%
    Time dependence of the metrological sensitivity for displaced squeezed states
    $D(\alpha(r))S(r)|0\rangle$ at fixed total energy $E$, with
    $|\alpha(r)|^2=E/\oscR-\cosh(2r)/2$ and $r\in\{0,1,3,5\}$.
    As $r$ increases, energy is redistributed from displacement to squeezing via
    Eq.~\eqref{eq:energy-constraint}.
    Main panel: QFI.  Top inset: homodyne CFI.  Bottom inset: heterodyne CFI.
    Parameters: $\gamma\approx9.901\times10^{-3}$, $\Lambda=100$, $\osc=1$, $\Omega_R \approx 0.99$,
    $T=10^{-3}$.
  }
  \label{fig:diff-r}
\end{figure}

A striking hierarchy emerges at short times.
Because the QFI starts from zero for every $r$, the relevant quantity is the rate
of its early build-up, and this is set by the suppressed position variance: with
$\sigma_{11}(0)=e^{-2r}/(2\oscR)$ a factor $e^{-2r}$ below the vacuum level, the
position-quadrature sensitivity $(\partial_T\sigma_{11})^2/\sigma_{11}^2$ is
enhanced by $e^{4r}$ as soon as the bath imprints a temperature dependence on
$\sigma_{11}$.
The strongly squeezed probe therefore reaches a sizeable sensitivity far earlier
than the less squeezed one and, in the non-Markovian dynamics, its first revival peak
stands well above the vacuum curve.
This head start is temporary: the interaction with the bath degrades the squeezing,
and at long times all curves converge to the common steady-state QFI, a property of
the equilibrium Gibbs state that carries no memory of the initial preparation.
Between these limits the exact QFI attains a maximum at a finite optimal time
$t^*(r)>0$, which establishes squeezing as a genuine thermometric resource and
which moves to earlier times as $r$ increases, since more strongly squeezed probes
must be read out before thermalisation erases their non-classical correlations.
This is the qualitative signature of non-Markovianity.
The Markovian QFI is qualitatively different.
A squeezed probe evolving under the GKLS dynamics also starts from $\QFI(0)=0$ and
rises monotonically to the same steady-state value (Proposition~\ref{prop:iso} and
Lemma~\ref{lem:asymp-mono}, with the approach always from below): there is no
interior maximum and no early-time optimum, so its optimum is again $t^*\to\infty$
and the squeezing can change the rate at which the Markovian QFI builds up, but it cannot create a finite-time optimum or outperform the stationary Markovian value.
The revivals of the exact dynamics, by contrast, create a genuine interior optimum
at finite $t^*(r)>0$ with no Markovian counterpart.
The two complementary resources therefore act in concert only out of the Markovian
regime: the squeezing $r$ sets how fast the transient sensitivity builds, and the
non-Markovian dynamics supplies the revival window in which that sensitivity can be
captured before it decays.
The measurement behaves as before: the homodyne Fisher information tracks the QFI
closely at every squeezing level, while the heterodyne curve falls increasingly far
below it as $r$ grows, since strong squeezing concentrates the temperature
information almost entirely in the position quadrature, where a balanced
two-quadrature measurement is wasteful.

Finally, we trace how the transient QFI builds up as the coupling is increased.
Figure~\ref{fig:diff-gammas} shows it on a log-log scale for
$\gamma\in\{10^{-6},\,10^{-4},\,2\times10^{-4},\,10^{-3},\,2\times10^{-3},\,
4.975\times10^{-3},\,9.901\times10^{-3}\}$ at fixed $T=10^{-3}$.

\begin{figure}[h]
  \centering
  \includegraphics[width=0.6\linewidth]{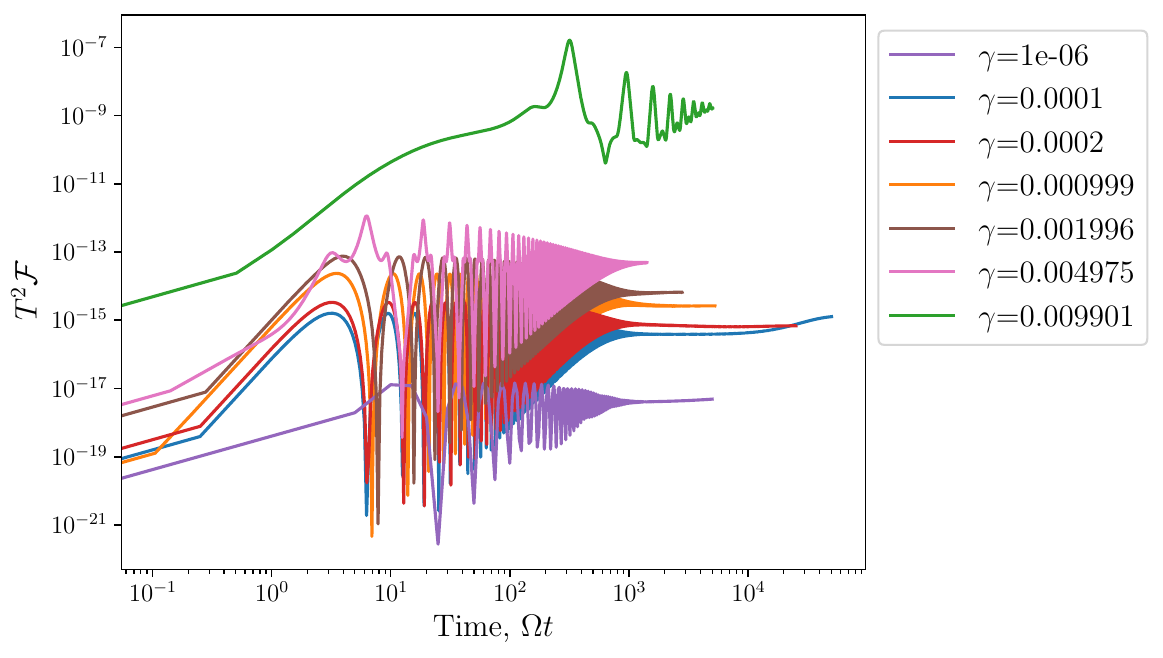}
  \caption{%
    Log-log plot of the QFI as a function of time for coupling strengths
    $\gamma$ (see legend), at $T=10^{-3}$,
    $\Lambda=100$, $\Omega=1$, and $\Omega_R \approx 0.99$.
    Initial state: ground state $(r=|\alpha|=0)$.
  }
  \label{fig:diff-gammas}
\end{figure}

At the weakest coupling $(\gamma=10^{-6})$ the probe is essentially Markovian, and
the QFI rises smoothly and monotonically with no appreciable oscillations.
Increasing $\gamma$ has two compounding effects.
The bath memory kernel~\eqref{eq:langevin} grows in importance, so that the
non-Markovian oscillations in the QFI set in earlier and with larger amplitude; and
the peak transient QFI itself rises by many orders of magnitude, exceeding the
weak-coupling value at the strongest coupling near the stability
boundary~\eqref{eq:stability}.
The latter enhancement can be traced to the renormalized frequency
$\oscR^2=\Omega^2-\gamma\Lambda$: as $\gamma\to\Omega^2/\Lambda$ one has
$\oscR\to 0$, so that even at $T=10^{-3}$ the probe operates in the regime
$T\gg\oscR$ where its thermal fluctuations are effectively classical and its
temperature sensitivity is correspondingly large.
On the log-log scale the curves illustrate the rapid enhancement of the transient
QFI as the stability boundary is approached. We do not assign a universal scaling
exponent to this transient envelope, since it depends on the finite-time window
and on the oscillatory revival structure.

\section{Discussion and conclusions}
\label{sec:discussion}

The superior performance of the exact solution comes from mode hybridisation.
Because the total Hamiltonian~\eqref{eq:H-total} is quadratic, the dynamics of the
combined probe--bath system decomposes into the eigenmodes of the whole, and these
eigenmodes are hybrid superpositions of the bare probe oscillator and the bath
oscillators.
It is in this hybridised basis that the problem is exactly solvable, and it is the
hybridised modes, not the bare probe, that the dressed frequency $\oscR$ and the
effective damping frequency $\effW$ describe; the Green's function
$G(t)=e^{-\gamma t/2}\sin(\effW t)/\effW$ encodes their structure in full.
A change in temperature alters the occupation of all of these hybridised modes, and
through the coupling the probe coordinate $q(t)$ inherits temperature sensitivity
from the bath modes as well as from its own.
This is the common origin of the two enhancements we report: the polynomial
low-temperature scaling of the steady-state QFI in Sec.~\ref{subsec:steady} and
the non-Markovian revivals of Sec.~\ref{subsec:transient}, both of which reflect
the fact that information flowing from the probe into the hybridised modes is not
lost but can return coherently at later times.
The Markovian master equation, by contrast, works only with the reduced state of
the probe.
Tracing out the bath compresses the entire hybridisation into the three effective
parameters $\oscR$ and $\gamma_{\uparrow(\downarrow)}$ and replaces the explicit hybridised probe--bath dynamics by an effective reduced
semigroup, leaving a description that can only reproduce the
exponential low-temperature suppression of a bare oscillator of gap $\oscR$.
The revivals and the polynomial scaling are genuine physical effects with no
representation within the master equation, so the gap between the two descriptions
is one of kind, not of degree.

This reading also clarifies how the product-state assumption in
Eq.~\eqref{eq:initial-state} should be understood. It is not meant to describe a
global equilibrium state of the strongly coupled probe--bath compound. Rather, it
defines an operational quench protocol: the probe is prepared independently, the
bath is initially thermal, and the probe--bath interaction is switched on at
$t=0$. The work cost and possible initial slip associated with this switch-on are
not included in the thermometric figure of merit; they belong to the preparation
stage of the protocol. Once the interaction is switched on, the subsequent
build-up of probe--bath correlations, hybridisation, and memory effects is fully
captured by the exact quadratic dynamics.

The transient regime is best understood in terms of two resources.
The first is the interrogation time.
In the Markovian case the QFI rises monotonically from zero to the steady-state
value (Appendix~\ref{app:monotonicity}), so its optimum is always at the boundary
$t^*\to\infty$, with no interior maximum---and this holds whether the probe starts
in the ground state or in a strongly squeezed state.
The exact dynamics overturns this conclusion, since the revivals render the QFI
non-monotonic and make a finite interior $t^*>0$ strictly optimal, so that time
becomes a resource to be deliberately exploited.
The second resource is squeezing.
An initial squeezed state concentrates the probe energy into a reduced
position-quadrature variance and so builds up a large sensitivity quickly; this
advantage is transient, since thermalisation erodes the squeezing, but in the
non-Markovian regime the revivals keep the sensitivity elevated for longer, so that
a strongly squeezed probe read out at a well-chosen time can far surpass any
coherent or thermal state of the same energy.
The two resources are complementary: squeezing sets how fast the transient
sensitivity builds, and the non-Markovian dynamics widens the window in which it can
be harvested.
The analogy with phase estimation is instructive---there squeezing yields a
Heisenberg-scaling advantage, here a transient thermometric one---but with the
essential difference that the resource is consumed by thermalisation, which is
exactly why a non-equilibrium protocol, measuring the probe before equilibrium is
reached, is indispensable for exploiting non-classical states.

It is illuminating to set these findings beside the fermionic companion
study~\cite{RavellRodriguez2024}, and the parallels run deeper than one might at
first expect. Both probes display a non-monotonic transient QFI with
non-Markovian revivals, and in both cases these revivals can generate a genuine
finite-time advantage over simply waiting for the steady state. Both systems also
escape the exponential Boltzmann suppression at strong coupling: the steady-state
relative error diverges polynomially in either case, with comparable
exponents---here $(\eta\approx2.03)$, close to the value found in the fermionic
setting. These features therefore appear to be robust consequences of strongly
coupled quadratic probe--bath dynamics, rather than of the statistics of the
probe alone.

Where the bosonic probe genuinely differs is in its initial-state resources. Its
infinite-dimensional Hilbert space admits squeezed preparations, which can
accelerate the early build-up of temperature sensitivity by suppressing the
position-quadrature variance. This squeezing-induced enhancement has no analogue
for a single fermionic mode, where a lone occupation number leaves no room for
such a resource. A squeezing advantage in Gaussian thermometry was demonstrated, in a
complementary weak-coupling setting, by Ref.~\cite{Mirkhalaf2024}; what is new
here is its interplay with non-Markovian revivals and with the strong-coupling
frequency renormalisation, together with the monotonicity result of
Appendix~\ref{app:monotonicity} that pins down why the Markovian description
misses the finite optimal time.

Our conclusion that strong coupling helps should be read with the regime
dependence recently emphasised for a qubit probe~\cite{Tan2025} in mind, and the
two pictures are in fact consistent. The steady-state enhancement we report is
an ultralow-temperature effect, operative precisely in the regime $(T\ll\oscR)$
where Ref.~\cite{Tan2025} likewise finds strong coupling beneficial; at moderate
temperatures the exact and Markovian relative errors in
Fig.~\ref{fig:shot-noise} converge, leaving little steady-state strong-coupling
gain. The transient advantage found here has a different origin. It is not a
simple consequence of improving the steady-state signal-to-noise ratio, but of
combining strong-coupling-induced non-Markovian revivals with two bosonic
features inacessible to a single qubit: the position-quadrature squeezing of the
initial state, and the collapse of the renormalized gap $(\oscR\to0)$ as
$(\gamma\to\Omega^2/\Lambda)$, which pushes the probe into an effectively
classical, highly temperature-sensitive regime. Thus, in the transient protocol,
the role of strong coupling is to open a finite interrogation window through
memory revivals, while squeezing and gap renormalisation determine how much
temperature sensitivity can be harvested within that window. In this sense, our
results refine rather than contradict the qubit analysis.

These considerations point naturally to circuit quantum electrodynamics (cQED) as
the most promising experimental setting~\cite{Lvov2025}, since our model---a
harmonic oscillator coupled through its position to a Drude--Ohmic bosonic
bath---maps directly onto it.
The probe oscillator is realised by a microwave resonator, or by the harmonic mode
of a transmon in its linear regime, and the bath by whatever dissipative element
carries the temperature one wishes to estimate.
Two cases stand out.
The most direct is a resistive, normal-metal element coupled to the resonator,
whose Johnson--Nyquist noise furnishes an Ohmic bath with a Drude cut-off $\Lambda$
set by the $RC$ roll-off of the circuit; thermometry of such resistors, and more
generally of quasiparticle baths and normal-metal--insulator--superconductor
junctions at millikelvin temperatures, is exactly the goal of the quantum
calorimetry programme pursued by Pekola's group and
others~\cite{Lvov2025}.
A more ambitious target is the phonon temperature of the substrate itself: the
acoustic modes of the silicon or sapphire substrate exchange energy with the
superconducting film through anharmonic or piezoelectric coupling and could play
the role of the bosonic bath, with the resonator as probe, although reaching and
controlling the requisite strong coupling---and engineering an effective
Drude--Ohmic phonon--resonator spectral density---remains an open experimental
problem.
In either realisation, the squeezed probe states required by our protocol can be
generated with Josephson parametric amplifiers and injected into the resonator, and
the fast switch-on underlying the product initial state can be implemented by a
flux pulse on a tunable coupler.
The chief practical difficulty is one of timing: because the optimal interrogation
time $t^*(r)$ shrinks with increasing squeezing and may fall below the
thermalisation time $1/\gamma$, the protocol calls for time-resolved homodyne
detection at nanosecond resolution, which is nonetheless within reach of current
cQED technology.

Several questions remain open, and we close by sketching the most compelling of
them, with the caveat that each carries a substantial computational price: even
within the exactly solvable Boyanovsky--Jasnow framework the second moments of the
probe are already given by integrals whose integrands span several printed pages
(see Appendix~\ref{app:exact}) and demand careful numerical evaluation, so that
every extension below will be more demanding still. The most immediate extension
is to enlarge the probe from a single oscillator to an array of (N) coupled
oscillators sharing one bath. Already for two oscillators, inter-probe
correlations and collective hybridised modes should reshape both the transient
dynamics and the steady-state scaling. The central question is whether such an
array can achieve a Heisenberg-like scaling of the QFI with (N), and whether
bath memory can amplify this scaling further. Remaining with a single mode, one may instead broaden the class of initial states.
We confined ourselves to Gaussian states because they are experimentally natural
and closed under the quadratic dynamics, but non-Gaussian preparations such as
Fock or cat states might offer larger transient gains, raising the problem of
identifying the optimal probe state at fixed energy. The measurement side invites
a parallel generalization: the homodyne and heterodyne schemes used here are
practical, but not optimal in general, and it would be valuable to characterise
the optimal transient measurement and to ask whether it can be realised by photon
counting or related techniques. The model itself can also be pushed in two
directions. Replacing the Drude--Ohmic density by a sub- or super-Ohmic one would
test how the bath exponent controls the polynomial low-temperature scaling, and
in particular the $(\eta=2s)$ relation argued in Sec.~\ref{subsec:steady}.
Replacing the position-like coupling by a dispersive or Susskind--Glogower-type
nonlinear interaction would instead test whether the hybridisation-induced
enhancement survives beyond the exactly solvable linear regime. Finally, one may
relax the product initial condition and ask how correlations established before
the protocol begins affect the achievable precision. Underlying all of these extensions is a broader question about the operational
role of bath memory in thermometry. In the parameter regime studied here, the
largest transient enhancements occur together with pronounced memory revivals:
information that has flowed from the probe into the bath returns at later times
and temporarily increases the temperature sensitivity. This interpretation is
consistent with earlier work relating non-Markovian dynamics to metrological
performance and with approaches based on quantum Fisher-information flow or
Fisher-information contraction as probes of non-Markovianity
\cite{Lu2010,Chin2012,Abiuso2023,Parlato2025}. It would therefore be natural to
make this connection quantitative in the present bosonic thermometry setting, for
example by comparing the transient QFI enhancement with standard
non-Markovianity measures such as the Breuer--Laine--Piilo trace-distance
measure or the Rivas--Huelga--Plenio divisibility measure
\cite{Breuer2009,Rivas2010}. Such a comparison would clarify whether the
temperature QFI considered here is merely correlated with bath memory, or whether
it can be turned into an operational witness of non-Markovian information
backflow.

In summary, a quantum harmonic oscillator strongly coupled to a Drude--Ohmic
bosonic bath supports two qualitatively distinct thermometric regimes, and both
favour strong coupling, with the caveats on regime and figure of merit noted above.
At equilibrium, the exponential Boltzmann suppression of the relative error gives
way to the polynomial divergence of Eq.~\eqref{eq:exact-scaling}, a direct
imprint of mode hybridisation.
Out of equilibrium, the bath memory produces QFI revivals that single out a
finite optimal interrogation time, and squeezed preparations contribute a large,
though perishable, transient advantage; the renormalized GKLS
description---monotone, with no interior optimum
(Appendix~\ref{app:monotonicity})---captures neither effect, and a plain homodyne
measurement remains near-optimal throughout.
The mechanism underlying all of these findings, the hybridisation of probe and
bath modes in a quadratic Hamiltonian, is universal and should persist across
bosonic platforms; and since the model maps directly onto circuit quantum
electrodynamics, the protocols analysed here appear within reach of current
superconducting-circuit technology.
We hope that these results will motivate non-equilibrium bosonic thermometry
experiments in superconducting circuits, and that they will stimulate further
theoretical work on the roles of non-Markovianity, hybridisation, and
non-classical resources in quantum sensing.

\begin{acknowledgments}
MW and MH are supported by the QuantERA II Programme (No 2021/03/Y/ST2/00178, acronym ExTRaQT) that has received funding from the European Union’s Horizon 2020 and from Polish National Science Center. MW
acknowledges grant PRELUDIUM-20 (grant number:
2021/41/N/ST2/01349) from the National Science Cen-
ter. 
RRR acknowledges support from the Government of Spain (Severo Ochoa
CEX2019-000910-S and FUNQIP), Fundaci\'o Cellex, Fundaci\'o Mir-Puig, and
Generalitat de Catalunya (CERCA program).
\end{acknowledgments}

\bibliographystyle{apsrev4-2}
\bibliography{references}

\appendix

\section{Exact solution: covariance matrix dynamics}
\label{app:exact}

This appendix collects the full analytical solution for the time evolution of the
probe's first and second moments, following the Schr\"odinger-picture construction
of Sec.~\ref{subsec:exact} and Ref.~\cite{BoyanoskyJasnow2017}.
We reproduce the central integrand in its original form, as it documents the scale
of the analytical and numerical effort underlying the figures of
Sec.~\ref{sec:results}.

The starting point is the retarded Green's function (implified form) of the damped probe,
\begin{align}
  G(t) \;=\; e^{-\gamma t/2}\,\frac{\sin(\effW t)}{\effW},
  \label{eq:green}
\end{align}
with $\effW = \sqrt{\oscR^2 - \gamma^2/4}$, which satisfies
\begin{align}
  G(0) = 0, \qquad \dot{G}(0) = 1, \qquad \ddot{G}(0) = 0,
  \label{eq:green-ic}
\end{align}
and whose derivatives are
\begin{align}
  \dot{G}(t) &= e^{-\gamma t/2}
    \left(\cos\effW t - \frac{\gamma}{2\effW}\sin\effW t\right),
  \label{eq:Gdot}\\
  \ddot{G}(t) &= -e^{-\gamma t/2}
    \left[\effW^2\,\frac{\sin\effW t}{\effW}
          + \gamma\!\left(\cos\effW t
          - \frac{\gamma}{2\effW}\sin\effW t\right)\right],
  \label{eq:Gddot}
\end{align}
yet the exact solution requires $\ddot{G}(0)=0$, hence the simplified for of the Green's function is insufficient for the exact solution, and subleading terms need to be incorporated. 
\begin{remark}[Finite-cutoff correction used in the numerics]
\label{rem:finite-cutoff-green}
The leading large-cutoff Green function~\eqref{eq:green} captures the
underdamped resonant dynamics but does not reproduce the short-time
curvature of the exact finite-memory dynamics: it gives
$\ddot G(0)=-\gamma$, whereas the memory kernel enforces $\ddot G(0)=0$.
In the numerics we therefore retain, in addition, the leading contribution
of the rapidly decaying Drude pole,
\begin{equation}
  G(t)
  = \frac{\gamma}{\Lambda^2}\,e^{-\Lambda t}
  + e^{-\gamma t/2}\,\frac{\sin(\effW t)}{\effW}.
  \label{eq:green-numerical}
\end{equation}
Although exponentially suppressed on the bath-correlation timescale
$\Lambda^{-1}$, this term restores the correct initial curvature,
$\ddot G(0)=0$, and thereby improves the short-time behaviour.
The price is a small mismatch with the exact initial data $G(0)=0$,
$\dot G(0)=1$,
\[
  G(0)=\frac{\gamma}{\Lambda^2},
  \qquad
  \dot G(0)=1-\frac{\gamma}{\Lambda},
\]
of order $\Lambda^{-2}$ and $\Lambda^{-1}$ respectively---higher-order
finite-cutoff corrections that we do not subtract.
\end{remark}

In terms of $G(t)$, the probe coordinate evolves as
\begin{align}
  q(t) \;=\; \dot{G}(t)\,q(0) + G(t)\,p(0)
        + \int_0^t ds\; G(t-s)\,\xi(s),
  \label{eq:q-solution}
\end{align}
so that, with the initial first moments $Q_0=\langle q(0)\rangle$ and
$P_0=\langle p(0)\rangle$ and using $\langle\!\langle\xi(t)\rangle\!\rangle=0$, the
mean values are
\begin{align}
  \langle q(t)\rangle &= \dot{G}(t)\,Q_0 + G(t)\,P_0,
  \label{eq:q-mean}\\
  \langle p(t)\rangle &= \ddot{G}(t)\,Q_0 + \dot{G}(t)\,P_0.
  \label{eq:p-mean}
\end{align}
For the ground state, and indeed for any squeezed vacuum, $Q_0=P_0=0$ and the
first moments vanish at all times.

The noise operator $\xi(t)$ has zero mean and the two-time correlation function
\begin{align}
  \langle\!\langle \xi(t_1)\,\xi(t_2)\rangle\!\rangle
  \;=\;
  \frac{1}{\pi}
  \int_{-\infty}^{+\infty} d\omega\; J(\omega)\,n(\omega)\,
  e^{i\omega(t_1-t_2)},
  \label{eq:xi-correlator}
\end{align}
with $n(\omega)=(e^{\omega/T}-1)^{-1}$ and $\langle\!\langle\,\cdot\,\rangle\!\rangle$
the average over the bath state $\pi_B$.
For the Drude--Ohmic density~\eqref{eq:spectral-Drude} the integrand has poles on
the imaginary axis and, at large $\Lambda$, evaluates by residues to a sum of
decaying exponentials.

The second moments then follow from Eq.~\eqref{eq:q-solution} and its time
derivative.
Writing $QQ_0=\Tr[q^2(0)\rho_S(0)]$, $PP_0=\Tr[p^2(0)\rho_S(0)]$, and
$PQ_0=\Tr[q(0)p(0)\rho_S(0)]$ for the initial second moments, one obtains
\begin{align}
  \langle q^2(t)\rangle
  &=\; G^2(t)\,PP_0 + G(t)\dot{G}(t)\,(2PQ_0 + i)
     + \dot{G}^2(t)\,QQ_0
  \notag\\
  &\quad+\;
  \frac{1}{\pi}\int_{-\infty}^{+\infty}d\omega\;
  J(\omega)\,n(\omega)
  \int_0^t \!dt_1\!\int_0^t \!dt_2\;
  G(t-t_1)\,G(t-t_2)\,e^{i\omega(t_1-t_2)},
  \label{eq:q2}
\end{align}
\begin{align}
  \langle p^2(t)\rangle
  &=\; \dot{G}^2(t)\,PP_0 + \dot{G}(t)\ddot{G}(t)\,(2PQ_0+i)
     + \ddot{G}^2(t)\,QQ_0
  \notag\\
  &\quad+\;
  \frac{1}{\pi}\int_{-\infty}^{+\infty}d\omega\;
  J(\omega)\,n(\omega)
  \int_0^t \!dt_1\!\int_0^t \!dt_2\;
  \dot{G}(t-t_1)\,\dot{G}(t-t_2)\,e^{i\omega(t_1-t_2)},
  \label{eq:p2}
\end{align}
\begin{align}
  \langle q(t)p(t)\rangle
  &=\; G(t)\dot{G}(t)\,PP_0
     + G(t)\ddot{G}(t)\,PQ_0
     + \dot{G}^2(t)\,(PQ_0+i)
     + \dot{G}(t)\ddot{G}(t)\,QQ_0
  \notag\\
  &\quad+\;
  \frac{1}{\pi}\int_{-\infty}^{+\infty}d\omega\;
  J(\omega)\,n(\omega)
  \int_0^t \!dt_1\!\int_0^t \!dt_2\;
  G(t-t_1)\,\dot{G}(t-t_2)\,e^{i\omega(t_1-t_2)}.
  \label{eq:qp}
\end{align}

The double time integrals can be reduced to single frequency integrals by
integration by parts.
Collecting the contributions to $\langle q^2(t)\rangle$ and
$\langle p^2(t)\rangle$ into the auxiliary integral
\begin{align}
  I(t)
  \;=\;
  \frac{1}{\pi}\int_{-\infty}^{+\infty}d\omega\; J(\omega)\,n(\omega)
  \Bigg[
    &\frac{1}{2\effW^2}\,e^{-\gamma t}\sin^2(\effW t)
  \notag\\
    +\;\frac{\omega^2+\oscR^2}{4\effW^2}\,e^{-\gamma t/2}
    \bigg|
    &t\,e^{it\effW/2}\,
      \mathrm{sinc}\!\left(\tfrac{\omega-\effW}{2}t-i\tfrac{\gamma}{4}t\right)
  \notag\\
    -\;
    &t\,e^{-it\effW/2}\,
      \mathrm{sinc}\!\left(\tfrac{\omega+\effW}{2}t-i\tfrac{\gamma}{4}t\right)
    \bigg|^2
  \Bigg],
  \label{eq:It}
\end{align}
the second moments are expressed through $I(t)$ and its limiting forms.
For the Drude--Ohmic density~\eqref{eq:spectral-Drude} the integrand of
Eq.~\eqref{eq:It} is an extraordinarily long closed-form expression. To give the
reader a sense of the scale of the analytic objects handled in this work, we
reproduce a representative portion below, in the variable $x=\omega$:

\small
\begin{align}
&2\Bigg(
\frac{e^{-t\Lambda}(\cosh t\Lambda - \cos tx)\,\gamma^2}
     {\Lambda^4(x^2+\Lambda^2)}
\notag\\
&+
\frac{e^{-\frac{t}{2}(\gamma+2\Lambda)}}
     {\effW\bigl(16\effW^4+8(\gamma^2-4x^2)\effW^2+(4x^2+\gamma^2)^2\bigr)
      \Lambda^2(x^2+\Lambda^2)}
\Big[
8\sin(t\effW)x^4
-8e^{\frac{t\gamma}{2}}\effW\sin(tx)x^3
\notag\\
&\quad
-4e^{t\Lambda}\gamma\sin(t\effW)\sin(tx)x^3
+8e^{t\Lambda}\Lambda\sin(t\effW)\sin(tx)x^3
+8e^{\frac{t\gamma}{2}+t\Lambda}\effW\gamma x^2
-8e^{\frac{t\gamma}{2}+t\Lambda}\effW\Lambda x^2
\notag\\
&\quad
-8\effW^2\sin(t\effW)x^2
+2\gamma^2\sin(t\effW)x^2
+4\gamma\Lambda\sin(t\effW)x^2
+8e^{\frac{t\gamma}{2}}\effW^3\sin(tx)x
+2e^{\frac{t\gamma}{2}}\effW\gamma^2\sin(tx)x
\notag\\
&\quad
-8e^{\frac{t\gamma}{2}}\effW\gamma\Lambda\sin(tx)x
-e^{t\Lambda}\gamma^3\sin(t\effW)\sin(tx)x
-4e^{t\Lambda}\effW^2\gamma\sin(t\effW)\sin(tx)x
\notag\\
&\quad
-8e^{t\Lambda}\effW^2\Lambda\sin(t\effW)\sin(tx)x
+2e^{t\Lambda}\gamma^2\Lambda\sin(t\effW)\sin(tx)x
+8e^{\frac{t\gamma}{2}+t\Lambda}\effW^3\Lambda
+2e^{\frac{t\gamma}{2}+t\Lambda}\effW\gamma^2\Lambda
\notag\\
&\quad
-2e^{\frac{t\gamma}{2}}\effW\bigl(4(\gamma-\Lambda)x^2+(4\effW^2+\gamma^2)\Lambda\bigr)\cos(tx)
+\gamma^3\Lambda\sin(t\effW)
+4\effW^2\gamma\Lambda\sin(t\effW)
\notag\\
&\quad
+e^{t\Lambda}\bigl(\effW^2(8x^2-4\gamma\Lambda)-(4x^2+\gamma^2)(2x^2+\gamma\Lambda)\bigr)\cos(tx)\sin(t\effW)
\notag\\
&\quad
-2\effW\cos(t\effW)\Big(\!-4\Lambda\effW^2-4x^2\gamma+4x^2\Lambda-\gamma^2\Lambda
-e^{t\Lambda}\bigl(-4(\gamma-\Lambda)x^2-(4\effW^2+\gamma^2)\Lambda\bigr)\cos(tx)
\notag\\
&\quad\qquad
+e^{t\Lambda}x\bigl(4\effW^2-4x^2+\gamma^2-4\gamma\Lambda\bigr)\sin(tx)\Big)
\Big]\gamma
\;+\;\cdots
\Bigg),
\label{eq:integrand}
\end{align}
\normalsize
Every term of this integrand has the same anatomy: a rational prefactor in $\omega$,
whose denominator combines the bath factor $\omega^2+\Lambda^2$ with the quartic
resonance denominator
$16\effW^4+8(\gamma^2-4\omega^2)\effW^2+(4\omega^2+\gamma^2)^2
=(2\effW-i\gamma-2\omega)(2\effW+i\gamma-2\omega)(2\effW-i\gamma+2\omega)(2\effW+i\gamma+2\omega)$,
multiplied by an exponential factor in $t$ (built from $e^{-\gamma t}$,
$e^{-\gamma t/2}$, and $e^{-\Lambda t}$ and their products) and by a product of two
oscillatory factors, one drawn from $\{\sin,\cos\}(\effW t)$ and one from
$\{\sin,\cos\}(\omega t)$. The denominator thus carries the four dressed-mode poles
$\omega=\pm\effW\pm i\gamma/2$ together with the bath poles $\omega=\pm i\Lambda$. The
bracket multiplying $\gamma$ above is a single representative block---a compact,
partially collected excerpt of the full expression. The ellipsis $\cdots$ is not
cosmetic: written out as a sum over these simple poles, the complete position-variance
integrand comprises, by direct count, \emph{eighty-eight} terms of exactly this
structure, organised by descending powers of $\gamma$ from $\gamma^6$ to $\gamma^0$;
the companion momentum-variance integrand is of the same character and comparable
length. Closed-form expressions of this size cannot be integrated efficiently within a
high-level computer-algebra system across the ranges of temperature and squeezing we
explore. We therefore generated them symbolically, transcribed the resulting frequency
integrands to C\texttt{++}/CUDA, and evaluated the integrals over $\omega$ numerically
on a GPU~\cite{WinczewskiCUDA2024}. We monitored the numerical accuracy by progressively increasing the
recursion depth of the adaptive $\omega$-quadrature until the resulting curves no
longer changed at the resolution of the figures.

The covariance-matrix elements entering the QFI~\eqref{eq:qfi-gaussian} are then
\begin{align}
  \sigma_{11}(t) &= \langle q^2(t)\rangle - \langle q(t)\rangle^2,
  \label{eq:sigma11}\\
  \sigma_{22}(t) &= \langle p^2(t)\rangle - \langle p(t)\rangle^2,
  \label{eq:sigma22}\\
  \sigma_{12}(t) &= \tfrac{1}{2}\langle q(t)p(t)+p(t)q(t)\rangle
                   - \langle q(t)\rangle\langle p(t)\rangle,
  \label{eq:sigma12}
\end{align}
and substituting Eqs.~\eqref{eq:q-mean}--\eqref{eq:qp} gives them in closed form as
functions of $t$, $T$, $\gamma$, $\Lambda$, and the initial moments
$QQ_0,PP_0,PQ_0$.
In the limit $t\to\infty$ the transient terms decay as $e^{-\gamma t/2}$, the double
integrals converge, and the asymptotic covariance matrix
$\boldsymbol{\sigma}^\mathrm{ss}$, obtained from Eq.~\eqref{eq:integrand} by
residues, fixes the steady-state QFI shown in Fig.~\ref{fig:shot-noise}.

Now, we move to the steady-state regime. The stationary second moments are obtained directly from the exact solution in the
long-time limit $t,t'\gg1/\gamma$, where only the noise term of
Eq.~\eqref{eq:q-solution} survives. Carrying out the frequency integral of the
resulting correlator for the Drude--Ohmic density, one finds, following
Ref.~\cite{BoyanoskyJasnow2017} [their Eqs.~(II.39)--(II.40)],
\begin{align}
  \sigma_{11}^\mathrm{ss}
  &= \langle q^2\rangle
   = \frac{1}{2\pi}\int_{-\infty}^{\infty} d\omega\;
     \frac{J(\omega)\,\coth(\omega/2T)}
          {(\omega^2-\oscR^2)^2+(\gamma\omega)^2},
  \label{eq:fdt-qq}\\
  \sigma_{22}^\mathrm{ss}
  &= \langle p^2\rangle
   = \frac{1}{2\pi}\int_{-\infty}^{\infty} d\omega\;
     \frac{\omega^2\,J(\omega)\,\coth(\omega/2T)}
          {(\omega^2-\oscR^2)^2+(\gamma\omega)^2},
  \label{eq:fdt-pp}\\
  \sigma_{12}^\mathrm{ss} &= 0,
  \label{eq:fdt-qp}
\end{align}
the last following by stationarity (no net phase-space flux). These are the fluctuation--dissipation correlators built from the retarded
susceptibility
\[
  \chi(\omega)=\bigl[\oscR^2-\omega^2-i\gamma\omega\bigr]^{-1},
\]
valid for $\Lambda\gg\oscR,\gamma$. Crucially, they depend on the renormalized
frequency $\oscR^2=\osc^2-\gamma\Lambda$---consistent with the static value
$\chi(0)=1/\oscR^2$---and not on the bare frequency $\osc$. Closing the contour in
the upper half-plane and summing the residues at the resonant poles
$\omega=i\gamma/2\pm\effW$ and at the bosonic Matsubara poles
\[
  \omega=2\pi i\,\ell T,\qquad \ell=1,2,\ldots,
\]
of $\coth(\omega/2T)$ yields the closed expressions of
Ref.~\cite{BoyanoskyJasnow2017} [their Eqs.~(II.41) and~(II.43)]: a finite sum of
resonant contributions plus a temperature-dependent Matsubara series. In particular,
at low temperature the thermal part of the stationary position variance scales as
\[
  \sigma_{11}^{\mathrm{ss}}(T)-\sigma_{11}^{\mathrm{ss}}(0)\sim T^2,
\]
which underlies the exponent $\eta=2$ discussed in Sec.~\ref{subsec:steady}.


\section{Markovian approximation: analytical solution}
\label{app:markov}

For comparison with the exact dynamics we solve the (renormalized) GKLS master
equation of Sec.~\ref{subsec:markov} in closed form. Writing it as
\begin{equation}
  \partial_t\rho
  =-i\big[\oscR\,a^\dagger a,\rho\big]
   +\gamma_\downarrow\,\mathcal D[a]\rho
   +\gamma_\uparrow\,\mathcal D[a^\dagger]\rho,
  \qquad
  \mathcal D[A]\rho=A\rho A^\dagger-\tfrac12\{A^\dagger A,\rho\},
  \label{eq:markov-me}
\end{equation}
with the ladder operators of Eq.~\eqref{eq:ladder} and the rates
\begin{equation}
  \gamma_\downarrow=\gamma\,[\bar n+1],
  \qquad
  \gamma_\uparrow=\gamma\,\bar n,
  \qquad
  \bar n \equiv \bar n(\oscR)=\frac{1}{e^{\oscR/T}-1},
  \label{eq:markov-rates}
\end{equation}
matching Eq.~\eqref{eq:rates} of the main text, so that the net relaxation rate
\begin{equation}
  \gamma=\gamma_\downarrow-\gamma_\uparrow
  \label{eq:markov-gamma}
\end{equation}
is independent of $T$ and equals the coupling constant of
Eq.~\eqref{eq:spectral-Drude}, in agreement with the exact decay $e^{-\gamma t}$ of
Appendix~\ref{app:exact}. The renormalization of the above equation follows the method in Ref.~\cite{WinczewskiAlicki2021}. We choose its
renormalized frequency $(\Omega_R)$ to coincide with that appearing in
the exact solution, allowing a direct comparison between the
two descriptions. Because the probe remains Gaussian, its state is fixed at all times by its
first moments and its covariance matrix, which we now compute in the Heisenberg
picture by evolving observables under the adjoint generator $\Lind^\dagger$.

The renormalized generator~\eqref{eq:markov-me} is the secular limit of the
time-local, refined weak-coupling (cumulant) generator of
Refs.~\cite{WinczewskiAlicki2021,Lobejko2024}, and it is this limit that renders the
closed-form solution transparent. The system coupling operator $q\propto a+a^\dagger$
decomposes into the eigenoperators $A(+\oscR)=a$ and $A(-\oscR)=a^\dagger$ of the
free Hamiltonian, carrying Bohr frequencies $\mp\oscR$. A general weak-coupling
generator pairs these as $A(\omega)\,\rho\,A^\dagger(\omega')$ with
$\omega,\omega'\in\{\pm\oscR\}$; the \emph{non-secular} pairings $\omega\neq\omega'$
generate anomalous terms $\propto a\rho a$ and $a^\dagger\rho a^\dagger$ that
oscillate as $e^{\pm 2i\oscR t}$ and couple the occupation to the squeezing sector.
Retaining only the secular pairings $\omega=\omega'$, as in
Eq.~\eqref{eq:markov-me}, removes this coupling, and the adjoint generator then
closes separately on each eigenoperator sector---which is precisely the structure we
exploit. 

With the Heisenberg dissipator
$\mathcal D^\dagger[A]\,O=A^\dagger O A-\tfrac12\{A^\dagger A,O\}$, the adjoint
generator
$\Lind^\dagger[O]=i[\oscR a^\dagger a,O]
+\gamma_\downarrow\,\mathcal D^\dagger[a]\,O
+\gamma_\uparrow\,\mathcal D^\dagger[a^\dagger]\,O$
(the Lamb shift absorbed into $\oscR)$ acts on the linear sector as
\begin{equation}
  \Lind^\dagger(a)=-\Big(i\oscR+\tfrac{\gamma}{2}\Big)a,
  \qquad
  \Lind^\dagger(a^\dagger)=\Big(i\oscR-\tfrac{\gamma}{2}\Big)a^\dagger,
  \label{eq:markov-adjoint-lin}
\end{equation}
and, using $aa^\dagger=1+a^\dagger a$, on the quadratic sector as
\begin{equation}
  \Lind^\dagger(a^2)=-\big(2i\oscR+\gamma\big)\,a^2,
  \quad
  \Lind^\dagger(a^\dagger a)=-\gamma\,a^\dagger a+\gamma_\uparrow,
  \quad
  \Lind^\dagger(a^{\dagger 2})=\big(2i\oscR-\gamma\big)\,a^{\dagger 2}.
  \label{eq:markov-adjoint-quad}
\end{equation}
Equivalently, the vector of quadratic monomials
$\mathbf O=(a^2,\,a^\dagger a,\,aa^\dagger,\,a^{\dagger 2})^{\!\top}$ obeys the closed
linear system $\dot{\mathbf O}=\mathsf M\,\mathbf O$, with
\begin{equation}
  \mathsf M=
  \begin{pmatrix}
    -(2i\oscR+\gamma) & 0 & 0 & 0\\[2pt]
    0 & -\gamma_\downarrow & \gamma_\uparrow & 0\\[2pt]
    0 & -\gamma_\downarrow & \gamma_\uparrow & 0\\[2pt]
    0 & 0 & 0 & 2i\oscR-\gamma
  \end{pmatrix}.
  \label{eq:markov-matrix}
\end{equation}
The block-diagonal form is the fingerprint of the secular limit: the anomalous
entries $a^2$ and $a^{\dagger 2}$ decouple from the occupation block
$\{a^\dagger a,\,aa^\dagger\}$, whereas the non-secular terms discarded above would
populate the off-diagonal blocks. The degenerate occupation block preserves the
canonical constraint $aa^\dagger-a^\dagger a=1$ and drives
$\langle a^\dagger a\rangle$ to its thermal value $\bar n=\gamma_\uparrow/\gamma$.
Integrating Eqs.~\eqref{eq:markov-adjoint-lin}--\eqref{eq:markov-adjoint-quad},
\begin{align}
  a(t)&=e^{-(i\oscR+\gamma/2)t}\,a,
  &
  a^2(t)&=e^{-(2i\oscR+\gamma)t}\,a^2,
  \notag\\[2pt]
  (a^\dagger a)(t)&=e^{-\gamma t}\,a^\dagger a+\bar n\big(1-e^{-\gamma t}\big),
  &
  a^{\dagger 2}(t)&=e^{(2i\oscR-\gamma)t}\,a^{\dagger 2},
  \label{eq:markov-ops}
\end{align}
together with the Hermitian conjugates of the first row; the occupation relaxes
monotonically to $\bar n$ at rate $\gamma$.

Inserting Eq.~\eqref{eq:markov-ops} into the quadratures~\eqref{eq:ladder}, the
first moments perform a damped phase-space rotation,
\begin{align}
  \langle q(t)\rangle
  &=e^{-\gamma t/2}\Big[\cos(\oscR t)\,Q_0
     +\tfrac{1}{\oscR}\sin(\oscR t)\,P_0\Big],
  \notag\\
  \langle p(t)\rangle
  &=e^{-\gamma t/2}\Big[\cos(\oscR t)\,P_0-\oscR\,\sin(\oscR t)\,Q_0\Big],
  \label{eq:markov-first}
\end{align}
with $Q_0=\langle q(0)\rangle$ and $P_0=\langle p(0)\rangle$; compactly
$\mathbf d_t=e^{-\gamma t/2}\,\mathcal R_{\oscR t}\,\mathbf d_0$, where
$\mathcal R_{\oscR t}$ is the symplectic rotation generated by the Hamiltonian part.
Being $T$-independent and decaying to $\mathbf 0$, they carry no thermometric
information.

The second moments follow from
$q^2=(a^2+a^{\dagger 2}+2a^\dagger a+1)/(2\oscR)$,
$p^2=\oscR(2a^\dagger a+1-a^2-a^{\dagger 2})/2$, and
$qp=\tfrac{i}{2}(1+a^{\dagger 2}-a^2)$. Denoting the initial second moments
$QQ_0=\langle q^2(0)\rangle$, $PP_0=\langle p^2(0)\rangle$, and
$PQ_0=\langle q(0)p(0)\rangle$, the anomalous expectations
$\langle a^2(t)\rangle=e^{-(2i\oscR+\gamma)t}\langle a^2(0)\rangle$ rotate at $2\oscR$
as they decay. For the rotationally symmetric preparations used for the Markovian
curves---ground, vacuum, and thermal states, for which
$\langle a^2(0)\rangle=0$---these terms drop and the second moments reduce to
\begin{align}
  \langle q^2(t)\rangle
  &= e^{-\gamma t}\,QQ_0
    +\frac{\coth(\oscR/2T)}{2\oscR}\big(1-e^{-\gamma t}\big)
  \xrightarrow[t\to\infty]{}\ \frac{\coth(\oscR/2T)}{2\oscR},
  \label{eq:markov-qq}\\
  \langle p^2(t)\rangle
  &= e^{-\gamma t}\,PP_0
    +\frac{\oscR\coth(\oscR/2T)}{2}\big(1-e^{-\gamma t}\big)
  \xrightarrow[t\to\infty]{}\ \frac{\oscR\coth(\oscR/2T)}{2},
  \label{eq:markov-pp}\\
  \langle q(t)p(t)\rangle
  &= e^{-\gamma t}\,PQ_0+\frac{i}{2}\big(1-e^{-\gamma t}\big)
  \xrightarrow[t\to\infty]{}\ \frac{i}{2},
  \label{eq:markov-qp}
\end{align}
where we used
$(\gamma_\uparrow+\gamma_\downarrow)/(\gamma_\downarrow-\gamma_\uparrow)=2\bar n+1
=\coth(\oscR/2T)$.

Subtracting the first moments through
$\sigma_{kl}=\tfrac12\langle\{R_k,R_l\}\rangle-\langle R_k\rangle\langle R_l\rangle$,
the general (lab-frame) covariance matrix is
\begin{equation}
  \boldsymbol\sigma(t)
  =e^{-\gamma t}\,\mathcal R_{\oscR t}\,\boldsymbol\sigma(0)\,
   \mathcal R_{\oscR t}^{\!\top}
   +\big(1-e^{-\gamma t}\big)\,\boldsymbol\sigma_\infty,
  \qquad
  \boldsymbol\sigma_\infty
  =\frac{\coth(\oscR/2T)}{2}
   \begin{pmatrix}1/\oscR & 0\\[2pt] 0 & \oscR\end{pmatrix},
  \label{eq:markov-lab}
\end{equation}
the steady state $\boldsymbol\sigma_\infty$ being the Gibbs covariance
matrix~\eqref{eq:CM-thermal} at the renormalized frequency. For the symmetric
preparations $\boldsymbol\sigma(0)$ is isotropic and commutes with
$\mathcal R_{\oscR t}$, so the rotation cancels and Eq.~\eqref{eq:markov-lab} is
already affine in $u:=e^{-\gamma t}\in[0,1]$. For an anisotropic (squeezed)
preparation the rotation does not cancel and the lab-frame covariance oscillates at
$2\oscR$; passing to the interaction picture
$\tilde{\boldsymbol\sigma}(t):=\mathcal R_{\oscR t}^{\!\top}\,\boldsymbol\sigma(t)\,
\mathcal R_{\oscR t}$ and using
$\mathcal R_{\oscR t}^{\!\top}\,\boldsymbol\sigma_\infty\,\mathcal R_{\oscR t}
=\boldsymbol\sigma_\infty$ restores it,
\begin{equation}
  \boldsymbol\sigma(t)
  =e^{-\gamma t}\,\boldsymbol\sigma(0)
   +\big(1-e^{-\gamma t}\big)\,\boldsymbol\sigma_\infty,
  \label{eq:markov-affine}
\end{equation}
understood in the lab frame for symmetric preparations and in the interaction
picture otherwise. Equation~\eqref{eq:markov-affine} is the affine structure on
which the monotonicity analysis of Appendix~\ref{app:monotonicity} rests.

Finally, recognising that the optimal measurement on the thermal steady state
$\boldsymbol\sigma_\infty$ is an energy measurement, its QFI is the energy variance
divided by $T^4$,
\begin{equation}
  \QFI_\infty=\frac{\operatorname{Var}(H_S)}{T^4}
  =\frac{\oscR^2}{T^4}\,\bar n(\bar n+1),
  \label{eq:markov-ss-qfi}
\end{equation}
in agreement with Eq.~\eqref{eq:markov-qfi}.

Before closing this appendix, a remark on the unitary part of
Eq.~\eqref{eq:markov-me} is in order. The free evolution
$a\mapsto e^{-i\oscR t}a$ is the symplectic rotation $\mathcal R_{\oscR t}$ of
$(q,p)$. It leaves the isotropic steady state $\boldsymbol\sigma_\infty$ invariant
and acts trivially on the rotationally symmetric preparations---ground, vacuum, and
thermal states, for which $\langle a^2(0)\rangle=0$---used for the Markovian curves;
for a general (squeezed) preparation it multiplies the anomalous part of
$\boldsymbol\sigma(0)$ by $e^{-2i\oscR t}$, but the Gaussian QFI is invariant under
symplectic congruence, so it is unaffected. We may therefore work with the
rotation-invariant content, which is exactly the affine
form~\eqref{eq:markov-affine}.

\section{Closed Markovian Fisher informations and their monotonicity}
\label{app:monotonicity}

Under the renormalized GKLS dynamics of Eq.~\eqref{eq:GKLS} the temperature enters
the probe only through its covariance matrix. We show that the Markovian QFI rises
monotonically from $\QFI(0)=0$ to its stationary value $\QFI_\infty$ with no interior
optimum, so that the Markovian optimum is always the boundary $t^*\to\infty$. The
statement is global and elementary for the rotationally symmetric preparations used
in all the figures (Proposition~\ref{prop:iso}); for an arbitrary squeezed
preparation we prove the approach to $\QFI_\infty$ from below
(Lemma~\ref{lem:asymp-mono}) and global monotonicity (Proposition~\ref{prop:global}).
The same closed forms carry over to the practical homodyne and heterodyne
measurements (Propositions~\ref{prop:markov-homodyne} and~\ref{prop:markov-heterodyne}),
for which the absence of a finite-time optimum also holds. Throughout, it is the
single real relaxation rate $\gamma$ that forbids the oscillatory revivals which make
the exact dynamics non-monotonic.

We first reduce the QFI to a single elementary function of time. Since the
first-moment vector $\mathbf d_t=e^{-\gamma t/2}\mathcal R_{\oscR t}\mathbf d_0$ of
Eq.~\eqref{eq:markov-first} is $T$-independent, $\partial_T\mathbf d_t=\mathbf 0$ and
the Gaussian QFI~\eqref{eq:qfi-gaussian} reduces to its covariance part,
\begin{equation}
  \QFI(t)=2\Tr\!\big[(C^{(2)}_t\,\boldsymbol\sigma_t)^2
                  +\big(\tfrac12\,C^{(2)}_t\,\sympform\big)^2\big],
  \label{eq:qfi-cov}
\end{equation}
where $C^{(2)}_t$ solves the Sylvester equation~\eqref{eq:C2-eq}. The QFI is invariant
under symplectic congruence---both traces are unchanged when
$\boldsymbol\sigma\mapsto\mathcal R\boldsymbol\sigma\mathcal R^{\!\top}$,
$C^{(2)}\mapsto\mathcal R^{-\!\top}C^{(2)}\mathcal R^{-1}$---so we pass to the
interaction picture, in which Eq.~\eqref{eq:markov-affine} renders the covariance
affine in $u:=e^{-\gamma t}\in[0,1]$,
\begin{equation}
  \boldsymbol\sigma_t=\boldsymbol\sigma_\infty+u\,\boldsymbol\delta,
  \qquad
  \boldsymbol\delta=\boldsymbol\sigma(0)-\boldsymbol\sigma_\infty,
  \qquad
  \partial_T\boldsymbol\sigma_t=(1-u)\,\partial_T\boldsymbol\sigma_\infty,
  \label{eq:cm-affine}
\end{equation}
using $\partial_T\boldsymbol\sigma(0)=\mathbf 0$. A fixed symplectic rescaling
diagonalises the steady state to $\boldsymbol\sigma_\infty=(\bar n+\tfrac12)\,\mathds{1}$,
where
\begin{equation}
  \bar n=\frac{1}{e^{\oscR/T}-1}
  \qquad\text{and}\qquad
  \dot{\bar n}:=\partial_T\bar n=\frac{\oscR}{T^2}\,\bar n(\bar n+1)>0
  \label{eq:nbar-recall}
\end{equation}
are the equilibrium occupation of the renormalized mode and its temperature
derivative, so that $\partial_T\boldsymbol\sigma_\infty=\dot{\bar n}\,\mathds{1}$,
while a (possibly squeezed) initial state gives
$\boldsymbol\sigma_t=\operatorname{diag}(a_t,b_t)$ with $a_t,b_t$ affine in $u$.

\begin{lemma}[Closed-form Markovian QFI]\label{lem:closed}
For a single-mode Gaussian preparation evolving under Eq.~\eqref{eq:GKLS}, in the
rescaled interaction picture above with $\boldsymbol\sigma_t=\operatorname{diag}(a_t,b_t)$,
the quantum Fisher information is
\begin{equation}
  \QFI(t)=\frac{4\,\dot{\bar n}^{\,2}\,(1-u)^2\,
                \bigl[\,2(a_t^2+b_t^2)+1\,\bigr]}
               {16\,(a_t b_t)^2-1}\,,
  \qquad u=e^{-\gamma t}.
  \label{eq:qfi-closed}
\end{equation}
\end{lemma}
\begin{proof}
Write $\boldsymbol\sigma=\operatorname{diag}(a,b)$, $C^{(2)}=\begin{psmallmatrix}C_{11}&C_{12}\\C_{12}&C_{22}\end{psmallmatrix}$,
$\sympform=\begin{psmallmatrix}0&1\\-1&0\end{psmallmatrix}$, and $\partial_T\boldsymbol\sigma=\dot{\bar n}(1-u)\,\mathds{1}$.
A short matrix multiplication gives
$2\boldsymbol\sigma C^{(2)}\boldsymbol\sigma=\begin{psmallmatrix}2a^2C_{11}&2ab\,C_{12}\\2ab\,C_{12}&2b^2C_{22}\end{psmallmatrix}$
and
$\tfrac12\sympform C^{(2)}\sympform=\tfrac12\begin{psmallmatrix}-C_{22}&C_{12}\\C_{12}&-C_{11}\end{psmallmatrix}$,
so the Sylvester equation~\eqref{eq:C2-eq} reads, entry by entry,
\begin{align}
  \bigl(2ab+\tfrac12\bigr)C_{12}&=0,\nonumber\\
  2a^2C_{11}-\tfrac12 C_{22}&=\dot{\bar n}(1-u),\qquad
  2b^2C_{22}-\tfrac12 C_{11}=\dot{\bar n}(1-u).
  \label{eq:sylv-components}
\end{align}
Since $2ab+\tfrac12>0$ the first line forces $C_{12}=0$. The remaining two equations are
a linear system $\begin{psmallmatrix}2a^2&-\tfrac12\\-\tfrac12&2b^2\end{psmallmatrix}\!\begin{psmallmatrix}C_{11}\\C_{22}\end{psmallmatrix}=\dot{\bar n}(1-u)\begin{psmallmatrix}1\\1\end{psmallmatrix}$
with determinant $\Delta=4a^2b^2-\tfrac14$; by Cramer's rule
\begin{equation}
  C_{11}=\frac{\dot{\bar n}(1-u)\,(2b^2+\tfrac12)}{\Delta},\qquad
  C_{22}=\frac{\dot{\bar n}(1-u)\,(2a^2+\tfrac12)}{\Delta}.
  \label{eq:C-solution}
\end{equation}
With $C_{12}=0$ the two traces in Eq.~\eqref{eq:qfi-cov} are diagonal: $C^{(2)}\boldsymbol\sigma=\operatorname{diag}(aC_{11},bC_{22})$ gives
$\Tr[(C^{(2)}\boldsymbol\sigma)^2]=a^2C_{11}^2+b^2C_{22}^2$, while
$\tfrac12 C^{(2)}\sympform=\tfrac12\begin{psmallmatrix}0&C_{11}\\-C_{22}&0\end{psmallmatrix}$ squares to $-\tfrac14 C_{11}C_{22}\,\mathds{1}$, so
$\Tr[(\tfrac12 C^{(2)}\sympform)^2]=-\tfrac12 C_{11}C_{22}$. Hence
\begin{equation}
  \QFI=2a^2C_{11}^2+2b^2C_{22}^2-C_{11}C_{22}.
  \label{eq:qfi-traces}
\end{equation}
Inserting Eq.~\eqref{eq:C-solution} and pulling out the common factor
$\dot{\bar n}^2(1-u)^2/\Delta^2$, the bracket
$B:=2a^2(2b^2+\tfrac12)^2+2b^2(2a^2+\tfrac12)^2-(2b^2+\tfrac12)(2a^2+\tfrac12)$
expands and factorises as $B=(16a^2b^2-1)\,[\,2(a^2+b^2)+1\,]/4$. Using
$\Delta^2=(16a^2b^2-1)^2/16$, one factor of $16a^2b^2-1$ cancels and
Eq.~\eqref{eq:qfi-traces} collapses to Eq.~\eqref{eq:qfi-closed}.
\end{proof}

Because $a_t,b_t$ are affine in $u$, Eq.~\eqref{eq:qfi-closed} is an explicit
elementary function of time, free of any matrix inversion; this is what makes it the
numerically reliable way to evaluate the Markovian QFI at ultralow temperatures
(cf.\ the note following Fig.~\ref{fig:Mark-diff-temp}).

\begin{proposition}[Symmetric preparations]\label{prop:iso}
For a ground, vacuum, or thermal initial state the covariance is isotropic,
$a_t=b_t=\tfrac12+(1-u)\bar n$, and Eq.~\eqref{eq:qfi-closed} reduces to
\begin{equation}
  \QFI(t)=\frac{(1-u)\,\dot{\bar n}^{\,2}}
               {\bar n\,\bigl[\,1+(1-u)\bar n\,\bigr]}\,,
  \qquad u=e^{-\gamma t}.
  \label{eq:qfi-iso}
\end{equation}
Then $\QFI(t)$ increases strictly and monotonically from $\QFI(0)=0$ to
$\QFI_\infty=\dot{\bar n}^{\,2}/[\bar n(\bar n+1)]$ as $t$ runs from $0$ to $\infty$,
with no interior optimum; the optimal interrogation time is $t^*\to\infty$.
\end{proposition}
\begin{proof}
For the vacuum/ground state $\boldsymbol\sigma(0)=\tfrac12\mathds{1}$, so
Eq.~\eqref{eq:cm-affine} gives $a_t=b_t=(1-u)(\bar n+\tfrac12)+u\cdot\tfrac12
=\tfrac12+(1-u)\bar n$. Setting $a=b$ in Eq.~\eqref{eq:qfi-closed}, the two invariants
simplify to $2(a^2+b^2)+1=4a^2+1$ and $16(ab)^2-1=16a^4-1=(4a^2-1)(4a^2+1)$, so the
common factor $4a^2+1$ cancels and
\begin{equation}
  \QFI(t)=\frac{4\,\dot{\bar n}^2(1-u)^2}{4a^2-1}\,.
  \label{eq:qfi-iso-step}
\end{equation}
Writing $w:=1-u$, we have $a=\tfrac12+w\bar n$, hence
$4a^2-1=(1+2w\bar n)^2-1=4w\bar n(1+w\bar n)$. Substituting into
Eq.~\eqref{eq:qfi-iso-step} and cancelling one power of $w$ yields
Eq.~\eqref{eq:qfi-iso}, $\QFI=(\dot{\bar n}^2/\bar n)\,w/(1+w\bar n)$. Differentiating,
\begin{equation}
  \frac{d\QFI}{dw}=\frac{\dot{\bar n}^2}{\bar n}\,
  \frac{(1+w\bar n)-w\bar n}{(1+w\bar n)^2}
  =\frac{\dot{\bar n}^2}{\bar n\,(1+w\bar n)^2}>0.
  \label{eq:qfi-iso-deriv}
\end{equation}
Finally $w=1-e^{-\gamma t}$ increases strictly with $t$ (from $w=0$ at $t=0$ to $w=1$
as $t\to\infty)$, so $\QFI(t)$ does too; the endpoint values are $\QFI=0$ at $w=0$ and
$\QFI_\infty=\dot{\bar n}^2/[\bar n(\bar n+1)]$ at $w=1$. There is no interior
stationary point, so $t^*\to\infty$.
\end{proof}

This covers every figure in the paper. For a squeezed preparation the steady state
remains isotropic but $a_t\neq b_t$. Removing the squeezing angle by a
temperature-independent rotation, we may take the initial covariance
$\boldsymbol\sigma(0)=\tfrac12\operatorname{diag}(e^{-2r},e^{2r})$; written in the
relaxation variable $w:=1-u=1-e^{-\gamma t}$, Eq.~\eqref{eq:cm-affine} then gives the
distinct affine entries
\begin{equation}
  a(w)=\tfrac12(1-w)e^{-2r}+w\nu_0,\qquad
  b(w)=\tfrac12(1-w)e^{2r}+w\nu_0,\qquad
  \nu_0=\bar n+\tfrac12.
  \label{eq:squeezed-ab}
\end{equation}
We first record the approach to the steady state (Lemma~\ref{lem:asymp-mono}) and then
establish global monotonicity (Proposition~\ref{prop:global}).

\begin{lemma}[Approach to the steady state]\label{lem:asymp-mono}
For any single-mode Gaussian preparation at $T>0$, the closed form~\eqref{eq:qfi-closed}
expands about $u=0$ as
\begin{equation}
  \QFI(t)=\QFI_\infty-\mathcal A\,e^{-\gamma t}+O\!\big(e^{-2\gamma t}\big),
  \qquad
  \mathcal A=-\,\frac{d\QFI}{du}\Big|_{u=0},
  \label{eq:f-expansion}
\end{equation}
and, for the squeezed-vacuum preparation~\eqref{eq:squeezed-ab},
\begin{equation}
  \mathcal A
  =\frac{\dot{\bar n}^{\,2}\,\bigl(\bar n\cosh 2r+\sinh^2 r\bigr)}
        {\bar n^{2}(\bar n+1)^{2}}\;>\;0
  \qquad\text{for all } r\ge 0,\ \bar n>0.
  \label{eq:Avalue}
\end{equation}
Hence $\QFI(t)$ approaches $\QFI_\infty$ from below, never from above, and the optimum
is $t^*\to\infty$.
\end{lemma}
\begin{proof}
At $T>0$ the steady state is mixed,
$\det\boldsymbol\sigma_\infty=(\bar n+\tfrac12)^2>\tfrac14$, so the denominator
$16(a_tb_t)^2-1$ of Eq.~\eqref{eq:qfi-closed} is positive at $u=0$ and the
expansion~\eqref{eq:f-expansion} is regular. Write
$\QFI=4\dot{\bar n}^2(1-u)^2 f(u)$ with $f=N/D$, $N=2(a_t^2+b_t^2)+1$,
$D=16(a_tb_t)^2-1$; then
$\mathcal A=-\,d\QFI/du|_{0}=-4\dot{\bar n}^2\bigl[-2f(0)+f'(0)\bigr]$.
In the variable $u=1-w$ the entries of Eq.~\eqref{eq:squeezed-ab} read
$a_t=(1-u)\nu_0+\tfrac{u}{2}e^{-2r}$ and $b_t=(1-u)\nu_0+\tfrac{u}{2}e^{2r}$, so at
$u=0$ one has $a_0=b_0=\nu_0$ and the derivatives
$a_0'+b_0'=\tfrac12(e^{2r}+e^{-2r})-2\nu_0=\cosh 2r-(2\bar n+1)=:c$. Hence
$N(0)=4\nu_0^2+1$, $D(0)=(4\nu_0^2-1)(4\nu_0^2+1)$, giving $f(0)=1/(4\nu_0^2-1)$, and
$N'(0)=4\nu_0(a_0'+b_0')=4\nu_0 c$, $D'(0)=32\nu_0^3 c$, so
\begin{equation}
  f'(0)=\frac{N'(0)D(0)-N(0)D'(0)}{D(0)^2}
       =-\,\frac{4\nu_0 c}{(4\nu_0^2-1)^2}.
  \label{eq:fprime}
\end{equation}
Assembling,
$\mathcal A=4\dot{\bar n}^2\bigl[\,2(4\nu_0^2-1)+4\nu_0 c\,\bigr]/(4\nu_0^2-1)^2$.
Using $4\nu_0^2-1=4\bar n(\bar n+1)$ and $4\nu_0 c=(4\bar n+2)(\cosh2r-2\bar n-1)$,
the square bracket equals $2\bigl[(2\bar n+1)\cosh2r-1\bigr]$, whence
$\mathcal A=\dot{\bar n}^2[(2\bar n+1)\cosh2r-1]/[2\bar n^2(\bar n+1)^2]$. Finally
$[(2\bar n+1)\cosh2r-1]/2=\bar n\cosh2r+\tfrac12(\cosh2r-1)=\bar n\cosh2r+\sinh^2 r$,
which gives Eq.~\eqref{eq:Avalue}. Both terms $\bar n\cosh2r$ and $\sinh^2r$ are
non-negative (and $\bar n\cosh2r>0$ for $\bar n>0)$, so $\mathcal A>0$ strictly.
\end{proof}

\begin{proposition}[Global monotonicity]\label{prop:global}
For every temperature $T>0$ and every squeezing $r\ge0$, the Markovian $\QFI(t)$ of the
squeezed preparation~\eqref{eq:squeezed-ab} is strictly increasing in $t$. Consequently
it has no interior optimum, and the optimal interrogation time is $t^*\to\infty$.
\end{proposition}
\begin{proof}
Since $w=1-e^{-\gamma t}$ increases strictly with $t$ from $0$ to $1$, it suffices to
prove $d\QFI/dw>0$ on $w\in(0,1]$. Abbreviate $n:=\bar n>0$, $C:=\cosh 2r$,
$h:=C-1\ge0$, and $\nu_0=n+\tfrac12$. From Eq.~\eqref{eq:squeezed-ab},
\begin{equation}
  a(w)\,b(w)-\tfrac14=w(1-w)\bigl(\nu_0 C-\tfrac12\bigr)+w^2 n(n+1),
  \label{eq:ab-positive}
\end{equation}
whose right-hand side is strictly positive on $w\in(0,1]$ (both $\nu_0 C-\tfrac12$ and
$n(n+1)$ are positive); hence $D(w):=16\,a(w)^2 b(w)^2-1>0$ there, and the closed
form~\eqref{eq:qfi-closed} is regular---the removable singularity at $w=0$ (pure
initial state) being fixed by $\QFI(0)=0$.

Differentiating Eq.~\eqref{eq:qfi-closed}, written with $(1-u)^2=w^2$, gives
\begin{equation}
  \frac{d\QFI}{dw}=\frac{8\,\dot{\bar n}^{\,2}\,w^2}{D(w)^2}\,P(w),
  \label{eq:qfi-derivative-P}
\end{equation}
where $P(w)$ is a quartic polynomial in $w$. Its sign is most transparent in the
Bernstein basis on $[0,1]$,
\begin{equation}
  P(w)=\sum_{k=0}^{4}\binom{4}{k}B_k\,w^k(1-w)^{4-k}.
  \label{eq:P-bernstein}
\end{equation}
Setting $A:=h(2n+1)+2n$ and $G:=2n^2+2n+1$, expansion and conversion from the monomial
to the Bernstein basis give
\begin{align}
  B_0&=4(h+1)^2A,\nonumber\\
  B_1&=2A(A+2),\nonumber\\
  B_2&=\tfrac43\bigl[\,G h^2+h(8n^3+16n^2+10n+3)+2n(4n^2+6n+3)\,\bigr],\nonumber\\
  B_3&=2G\bigl[\,G h^2+2G h+4n(n+1)\,\bigr],\nonumber\\
  B_4&=4G^2A.
  \label{eq:B-coefficients}
\end{align}
For $n>0$ and $h\ge0$ every quantity on the right is non-negative: $A>0$, $G>0$,
$h+1=\cosh 2r\ge1$, and the bracketed polynomials in $n$ carry only positive
coefficients; hence each $B_k>0$. The Bernstein functions
$\binom{4}{k}w^k(1-w)^{4-k}$ are non-negative on $[0,1]$, so $P(w)>0$ for $0\le w\le1$.
By Eq.~\eqref{eq:qfi-derivative-P}, $d\QFI/dw>0$ on $(0,1]$; therefore $\QFI(t)$ is
strictly increasing in $t$, with no interior optimum, and $t^*\to\infty$.
\end{proof}

Together, Propositions~\ref{prop:iso} and~\ref{prop:global} and
Lemma~\ref{lem:asymp-mono} establish the boundary character $t^*\to\infty$ of the
Markovian optimum used throughout Sec.~\ref{sec:results}.

\begin{remark}[Role of the initial state and contrast with the exact dynamics]
The sign of $\mathcal A$ in Eq.~\eqref{eq:Avalue} is the same for a low-information
start and for a strongly squeezed one: in both cases the Markovian QFI rises from zero
to the \emph{same} steady-state value $\QFI_\infty$, which depends only on the
equilibrium Gibbs state and not on the initial squeezing. Squeezing accelerates the
early build-up of the QFI---through the suppressed $\sigma_{11}(0)$ discussed in
Sec.~\ref{subsec:transient}---but confers no Markovian transient advantage, since
there is no interior time at which to stop. The distinction from the exact dynamics is
therefore not the mere existence of an interior optimum but the oscillatory revivals:
the exact QFI is non-monotonic, with a sequence of interior maxima at finite $t^*>0$
set by the period $2\pi/\effW$, which a single-rate Markovian semigroup cannot
produce. In this sense the present statement is sharper than the bound of
Ref.~\cite{SekatskiPerarnau2022}, which constrains $\QFI/t$ for general Markovian
dynamics but leaves the monotonicity of $\QFI(t)$ itself open.
\end{remark}

The same affine structure also gives closed expressions for the classical Fisher
information of the Gaussian measurements used in the main text. These are elementary
consequences of the Fisher information of Gaussian probability distributions, or
equivalently of the general Gaussian-measurement formula~\eqref{eq:cfi-gaussian} of
Sec.~\ref{subsec:gaussian}~\cite{Cenni2022}. We record them because they make explicit
that the absence of a finite-time optimum is not restricted to the optimal SLD
measurement: it also holds for the Markovian homodyne and heterodyne measurements,
provided the homodyne phase is chosen in the co-rotating principal-quadrature frame.

\begin{proposition}[Markovian homodyne Fisher information]
\label{prop:markov-homodyne}
In the rescaled interaction picture used throughout this appendix, with
$\boldsymbol\sigma_t=\operatorname{diag}\!\bigl(a(w),b(w)\bigr)$,
$\partial_T\boldsymbol\sigma_t=w\,\dot{\bar n}\,\mathds{1}$, and $w=1-e^{-\gamma t}$,
a homodyne measurement of the rotated quadrature
$q_\theta=q\cos\theta+p\sin\theta$ has outcome variance
$V_\theta(w)=a(w)\cos^2\theta+b(w)\sin^2\theta$ and classical Fisher information
\begin{equation}
  \CFI_{\mathrm{hom}}^{(\theta)}(t)
  =\frac{w^2\,\dot{\bar n}^{\,2}}{2\,V_\theta(w)^2}\,.
  \label{eq:homodyne-cfi-theta}
\end{equation}
Optimising over the local-oscillator phase gives
\begin{equation}
  \CFI_{\mathrm{hom}}^{\mathrm{opt}}(t)
  =\frac{w^2\,\dot{\bar n}^{\,2}}{2\,\min\{a(w),b(w)\}^2}\,.
  \label{eq:homodyne-cfi-opt}
\end{equation}
For the squeezed-vacuum preparation~\eqref{eq:squeezed-ab} and $r\ge0$ one has
$a(w)\le b(w)$, so that
\begin{equation}
  \CFI_{\mathrm{hom}}^{\mathrm{opt}}(t)
  =\frac{2\,w^2\,\dot{\bar n}^{\,2}}
        {\bigl[(1-w)e^{-2r}+w(2\bar n+1)\bigr]^2}\,.
  \label{eq:homodyne-cfi-squeezed}
\end{equation}
For every fixed co-rotating quadrature angle $\theta$---and in particular for the
phase-optimised homodyne measurement---$\CFI_{\mathrm{hom}}(t)$ is strictly increasing
for $t>0$.
\end{proposition}
\begin{proof}
In the interaction picture the homodyne outcome is a centred normal variable with
variance $V_\theta(w)$. Since $\partial_T a(w)=\partial_T b(w)=w\dot{\bar n}$, we have
$\partial_T V_\theta(w)=w\dot{\bar n}$; the Fisher information of a centred normal
distribution with parameter-dependent variance $V$ is $\CFI=(\partial_T V)^2/(2V^2)$
[the single-quadrature case of Eq.~\eqref{eq:cfi-gaussian}; see also~\cite{Cenni2022}],
which gives
Eq.~\eqref{eq:homodyne-cfi-theta}, and Eq.~\eqref{eq:cfi-homo} for the position
quadrature. Optimising over the local-oscillator phase minimises $V_\theta$, hence
$V_\theta=\min\{a,b\}$, proving Eq.~\eqref{eq:homodyne-cfi-opt}; for $r\ge0$ the
squeezed quadrature has the smaller variance, $a(w)\le b(w)$, and
Eq.~\eqref{eq:homodyne-cfi-squeezed} follows by inserting $a(w)$ from
Eq.~\eqref{eq:squeezed-ab}.

For monotonicity, write $V_\theta(w)=(1-w)V_\theta(0)+w\nu_0$ with $\nu_0=\bar n+\tfrac12$.
This is a convex combination of the positive variances $V_\theta(0)$ and $\nu_0$, hence
strictly positive on $[0,1]$, so its reciprocal is regular and
\[
  \CFI_{\mathrm{hom}}^{(\theta)}(w)
  =\frac{\dot{\bar n}^{\,2}}{2}
   \left[\frac{w}{V_\theta(0)+w\bigl(\nu_0-V_\theta(0)\bigr)}\right]^2 .
\]
The bracketed function has derivative
$\tfrac{d}{dw}\,w/[V_\theta(0)+w(\nu_0-V_\theta(0))]
=V_\theta(0)/[V_\theta(0)+w(\nu_0-V_\theta(0))]^2>0$, since $V_\theta(0)>0$; being
non-negative and vanishing only at $w=0$, its square is strictly increasing on
$(0,1]$. As $w(t)=1-e^{-\gamma t}$ is strictly increasing in $t$, so is
$\CFI_{\mathrm{hom}}^{(\theta)}(t)$.
\end{proof}

\begin{proposition}[Markovian heterodyne Fisher information]
\label{prop:markov-heterodyne}
In the same rescaled interaction picture, heterodyne detection projects onto
coherent states, adding to the probe the vacuum seed
$\boldsymbol\sigma^M=\boldsymbol\sigma_{\mathrm{vac}}$ of Eq.~\eqref{eq:cfi-hetero};
the rescaling that isotropises the steady state sends
$\boldsymbol\sigma_{\mathrm{vac}}\mapsto\tfrac12\mathds{1}$. The outcome is therefore
a centred two-dimensional Gaussian with covariance
\[
  \boldsymbol\Sigma_{\mathrm{het}}(w)
  =\boldsymbol\sigma_t+\boldsymbol\sigma_{\mathrm{vac}}
  =\boldsymbol\sigma_t+\tfrac12\mathds{1}
  =\operatorname{diag}\!\bigl(a(w)+\tfrac12,\,b(w)+\tfrac12\bigr),
\]
and its classical Fisher information is
\begin{equation}
  \CFI_{\mathrm{het}}(t)
  =\frac{w^2\,\dot{\bar n}^{\,2}}{2}
   \left[\frac{1}{\bigl(a(w)+\tfrac12\bigr)^2}
        +\frac{1}{\bigl(b(w)+\tfrac12\bigr)^2}\right].
  \label{eq:heterodyne-cfi}
\end{equation}
For the squeezed-vacuum preparation~\eqref{eq:squeezed-ab}, this becomes
\begin{align}
  \CFI_{\mathrm{het}}(t)
  =2\,w^2\,\dot{\bar n}^{\,2}\Bigg[
   &\frac{1}{\bigl[(1-w)e^{-2r}+w(2\bar n+1)+1\bigr]^2} +\frac{1}{\bigl[(1-w)e^{2r}+w(2\bar n+1)+1\bigr]^2}\Bigg].
  \label{eq:heterodyne-cfi-squeezed}
\end{align}
Moreover, $\CFI_{\mathrm{het}}(t)$ is strictly increasing for $t>0$, provided
$\dot{\bar n}\neq0$.
\end{proposition}

\begin{proof}
For a centred multivariate normal distribution with covariance
$\boldsymbol\Sigma(T)$, the Fisher information is the second term of
Eq.~\eqref{eq:cfi-gaussian} (see also~\cite{Cenni2022}),
\[
  \CFI
  =\tfrac12\Tr\!\bigl[(\boldsymbol\Sigma^{-1}\partial_T\boldsymbol\Sigma)^2\bigr].
\]
Here
\[
  \boldsymbol\Sigma_{\mathrm{het}}(w)
  =\operatorname{diag}\!\bigl(a(w)+\tfrac12,\,b(w)+\tfrac12\bigr),
  \qquad
  \partial_T\boldsymbol\Sigma_{\mathrm{het}}(w)=w\,\dot{\bar n}\,\mathds{1},
\]
and substitution gives Eq.~\eqref{eq:heterodyne-cfi};
Eq.~\eqref{eq:heterodyne-cfi-squeezed} follows by inserting
Eq.~\eqref{eq:squeezed-ab}. It remains to check monotonicity. Each term of
Eq.~\eqref{eq:heterodyne-cfi} has the form
\[
  \left[\frac{w}{c+wd}\right]^2,
\]
with $c=a(0)+\tfrac12>0$, $d=\nu_0-a(0)$ (or the same with $a$ replaced by $b)$,
where $\nu_0=\bar n+\tfrac12$. The denominator is a convex combination of positive
numbers,
\[
  c+wd=(1-w)\bigl[a(0)+\tfrac12\bigr]+w\bigl[\nu_0+\tfrac12\bigr]>0
  \qquad(0\le w\le1),
\]
and
\[
  \frac{d}{dw}\,\frac{w}{c+wd}=\frac{c}{(c+wd)^2}>0 .
\]
Thus each term in Eq.~\eqref{eq:heterodyne-cfi} increases strictly with $w$. Since
$w(t)=1-e^{-\gamma t}$ is strictly increasing in $t$, the heterodyne Fisher
information is strictly increasing in time.
\end{proof}

\begin{remark}[Co-rotating and laboratory homodyne phases]
The homodyne monotonicity result refers to a quadrature with fixed angle in the
interaction picture, equivalently to a local-oscillator phase that follows the free
Markovian rotation at frequency $\oscR$. A homodyne measurement at a fixed laboratory
quadrature also probes the trivial rotation of the covariance ellipse, and its Fisher
information need not be monotone pointwise in time. This convention does not affect
either the phase-optimised homodyne Fisher information or the heterodyne Fisher
information.
\end{remark}
\end{document}